\documentclass[runningheads]{llncs}
\usepackage{amsmath,amssymb}
\usepackage{subfigure}
\usepackage[linesnumbered,ruled,noend]{algorithm2e}
\usepackage[noend]{algorithmic}
\usepackage{listings}
\usepackage{threeparttable}
\usepackage[T1]{fontenc}
\usepackage{pgfplots}
\usepackage{tikz}
\usepackage{float}
\usepackage{enumitem}
\usepackage{color}
\usepackage{cite}
\usepackage{hyperref}
\usetikzlibrary{patterns}
\pgfplotsset{compat=1.18}

\newcommand{\nop}[1]{}

\begin{document}


\title{Efficient Cloud-edge Collaborative Approaches to SPARQL Queries over Large RDF graphs} 

\author{Shidan Ma\inst{1}\and
		Peng Peng\inst{1}\and
		Xu Zhou\inst{1}\and
		M. Tamer {\"O}zsu\inst{2}\and
		Lei Zou\inst{3}\and
		Guo Chen\inst{1}}
\authorrunning{S. M et al.}
%
\institute{Hunan University, China \\
\email{\{msd673, hnu16pp, zhxu, guochen\}@hnu.edu.cn}
\and
University of Waterloo, Canada\\
\email{tamer.ozsu@uwaterloo.ca}
 \and
Peking University, China\\
\email{zoulei@pku.edu.cn}}

\maketitle
\begin{abstract}
With the increasing use of RDF graphs, storing and querying such data using SPARQL is a critical problem. Current solutions rely on cloud-based data management architectures, but often suffer from performance bottlenecks when bandwidth is limited or system load is high. 
We explore, for the first time, an edge computing solution to improve query performance. This approach requires offloading query processing to edge servers, which involves addressing two challenges: data localization and network scheduling.
Data localization requires computing the subgraphs maintained on edge servers and to quickly identify the servers that can handle specific queries. To address this challenge, we introduce a new concept of pattern-induced subgraphs. 
Network scheduling involves efficiently assigning queries to edge and cloud servers to optimize overall system performance. 
We tackle this by constructing a system model and formulate query assignment and resource allocation as a Mixed Integer Nonlinear Programming (MINLP) problem, which we solve using a modified branch-and-bound algorithm.
Experiments under a real cloud platform demonstrate that our proposed method outperforms the state-of-the-art baseline methods in terms of efficiency. The codes are available on GitHub\footnote{
https://github.com/msd673/edgeComputing\_gurobi.git}.
\end{abstract}

\section{Introduction}\label{sec:Introduction}
 With the continuous development of the internet and data science, graph data models have become an essential means for representing and processing complex relational data. The Resource Description Framework (RDF) is a standard model for describing and representing resource information, where each triple $\langle s, p, o \rangle$ captures a relationship between resources. An RDF dataset is commonly viewed as a graph, where subjects and objects are vertices, and triples are edges. 
 SPARQL\footnote{https://www.w3.org/TR/sparql11-query/} is the standard query language for efficiently extracting useful information from RDF graphs. 
 A SPARQL query can be regarded as a query graph with variables, and answering it over an RDF dataset is equivalent to finding subgraph matches of the query graph within the RDF graph.
 
As the use of RDF graphs proliferates across many applications, cloud-based RDF management systems have become the mainstream solution for storing and querying RDF data \cite{DBLP:journals/vldb/AliSYHN22,10.14778/3151106.3151109,DBLP:journals/vldb/KaoudiM15,10.1145/2588555.2594535,DBLP:journals/fcsc/Ozsu16,DBLP:journals/tkde/HusainMMKT11}. For instance, Amazon Web Services offers Neptune \cite{ISWC2018:Neptune}, a graph database service that supports RDF and SPARQL in a fully managed cloud environment. 
Cloud-based SPARQL services enable users to submit queries from anywhere around the world; however, since users may be located in different geographic regions, cross-region data transmission may potentially result in significant query latency. 
To effectively reduce these latencies, some solutions have been proposed to offload computation tasks from the cloud to the client side \cite{DBLP:conf/sigmod/Hartig13,DBLP:journals/ws/VerborghSHHVMHC16,DBLP:conf/semweb/HartigLP17}. While this approach can reduce the computational load on the cloud, the client typically evaluates only individual triple patterns and must transmit a large number of intermediate results, which can still negatively impact query performance. 


{Edge computing has recently emerged as a promising paradigm to address these limitations by moving data processing closer to end users.
In real-world systems, these architectures have been widely deployed in various latency-sensitive scenarios. For instance, in industrial Internet systems \cite{qiu2020edge,LEPHUOC201625,8289317}, devices, sensors, and control units are naturally organized into explicit dependency structures that support functionalities such as state monitoring and anomaly detection. Similarly, in intelligent transportation systems \cite{10133894,5767845}, road networks and traffic entities are typically modeled as graphs, and the associated analysis tasks must be performed in a timely manner close to the data sources. 
}

This paper studies SPARQL query processing in an edge–cloud architecture, where the cloud collaborates with geographically distributed edge servers to handle end-user requests. Queries are submitted by end users and may be executed on edge servers or in the cloud. Due to the limited resources of edge servers, they can only store partial subgraphs of the complete RDF graph. Thus, it is important to tackle the challenge of determining which specific subgraphs should be stored on the edge servers and which queries should be executed by edge servers or sent to the cloud.
This paper focuses on a common deployment paradigm in which edge servers only collaborate with the cloud and not with each other \cite{TON2019:OnDisc,TMC2024:QoEAware,TMC2024:GameTheoretic}. When end users and edge servers are geographically separated, direct communication with the cloud is more efficient due to the high cost of inter-edge communication.

\begin{example}
Fig. \ref{fig:system_overview} shows a cloud–edge system with two edge servers (ESs) and five end users (EUs). The entire RDF graph is stored in the cloud, while pattern-induced subgraphs are stored on $ES_1$ and $ES_2$, respectively.
$EU_1$ and $EU_2$ connect only to $ES_1$, $EU_4$ and $EU_5$ only to $ES_2$, and $EU_3$ can connect to both.
End users issue SPARQL queries, which are processed by the associated edge servers or the cloud, and the results are returned to the users.
\end{example}

\begin{figure}
    \centering
      \includegraphics[width=0.85\columnwidth]{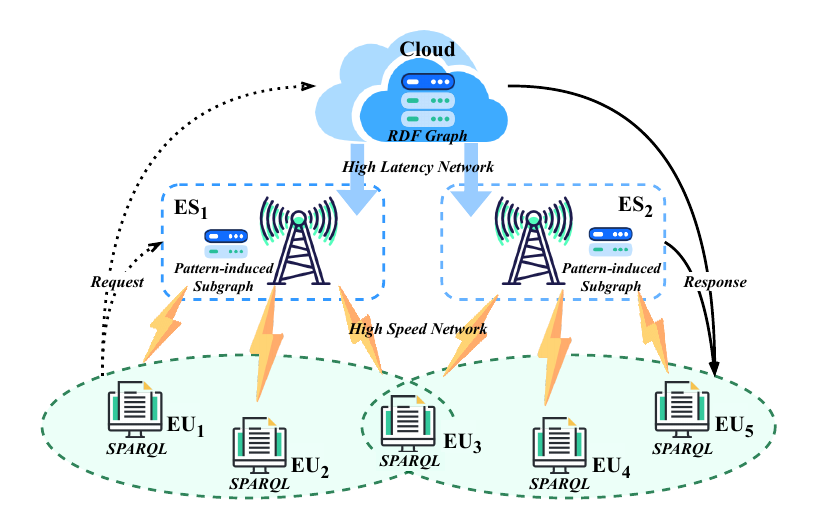}
\caption{Example Edge-Cloud System}%
 \label{fig:system_overview}
\end{figure}
\setlength{\textfloatsep}{0pt}

Real SPARQL workloads often exhibit strong locality, with queries submitted from a particular area during a given period frequently sharing similar structures \cite{DASFAA2018:FMQO,TKDE2021:FMQO}. These recurring query patterns can be identified and extracted. Given the finite resources of edge servers, each server can store only the subgraphs that match these frequently recurring patterns submitted by nearby users to improve query processing efficiency.

\begin{figure}
\centering
\subfigure[{Example RDF Graph}]{%
        \centering
		\includegraphics[width=0.9\columnwidth]{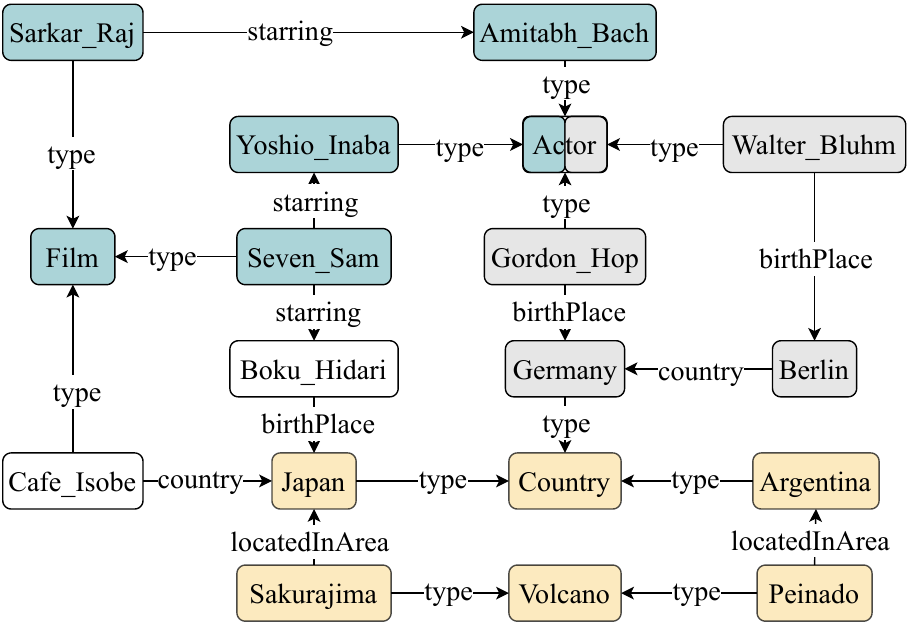}
       \label{fig:ExampleGraph}%
       }
   \subfigure[{{Example Workload}}]{%
        \centering
		\includegraphics[width=0.75\columnwidth]{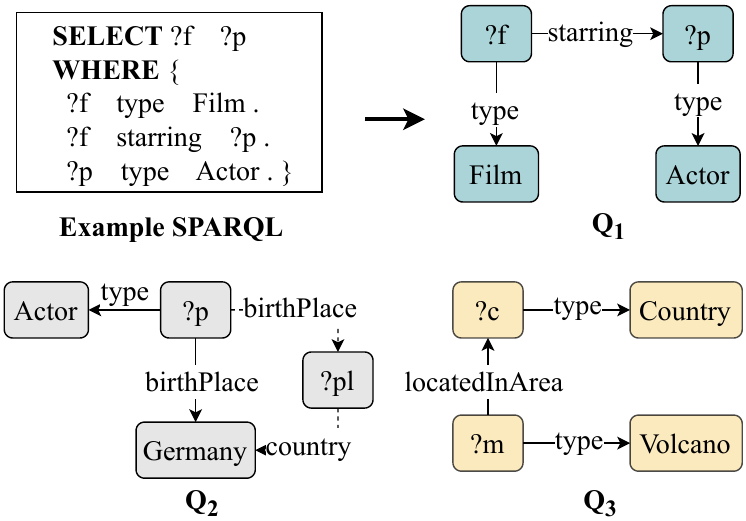}
       \label{fig:FirstExampleQuery}%
       }
\caption{Example RDF Graph and Workload}%
 \label{fig:ExampleGraphWorkload}
\end{figure}

To address this issue, based on recurring query patterns, this paper explores how to execute SPARQL queries in this edge-cloud architecture to reduce the latency of SPARQL query processing in the cloud. When a user near an edge server submits a query, the system first determines whether the query can be processed locally at the edge server. Even if it can, forwarding to the edge server is not automatic, and a further decision is made as to whether to execute the query at the edge server or at the cloud. This decision problem arises for a number of reasons, such as the current load on the edge server. The problem under consideration comprises two primary tasks: data localization and network scheduling. 

\emph{\underline{Data Localization.}} Given that each edge server has a limited capacity, it can only store some subgraphs of the complete RDF graph. Thus, it is crucial to determine which kinds of subgraphs should be maintained at the edge servers. This challenge falls within the scope of classical database problems related to data localization. 
In this paper, we introduce a novel concept of \emph{pattern-induced subgraphs}. These are subgraphs induced by the edges found within the matches of frequently queried patterns. Each edge server stores the pattern-induced subgraphs corresponding to several frequently queried patterns that it has previously seen. Thus, queries isomorphic to these patterns can be assuredly processed at the edge server. Thus, through the utilization of pattern-induced subgraphs, we present an effective approach for localizing queries.

\emph{\underline{Network Scheduling.}} As indicated above, even if a query can be executed at an edge server, a decision must be made whether to offload it to the edge or execute it at the cloud. Thus, we should design a method that allocates queries to appropriate computation nodes to enhance overall performance. This challenge is closely related to classical network scheduling problems. 
In this paper, we model the joint query assignment decision and computational resource allocation problem as a Mixed Integer Non-Linear Programming problem (MINLP). Each requirement mentioned above corresponds to a constraint within the MINLP. Solving the MINLP yields the optimal assignment of each query to an appropriate server.

\textbf{Contributions.} The main contributions of this paper are:

\begin{itemize}[left=0pt]
\item We introduce a novel challenge: the efficient offloading of SPARQL query processing in edge-cloud environments. This problem combines topics from both the database and the network domains, and as far as we know, this paper is the first to consider this particular problem (Section \ref{sec:Problem}).
\item For quickly determining which edge servers are capable of handling a specific query, we introduce a novel concept of \emph{pattern-induced subgraphs} to maintain some subgraphs of the entire RDF graph in the edge servers (Section \ref{sec:Problem}).
\item To appropriately allocate the queries submitted by end users to edge servers or the cloud, we model the joint formulation of query assignment and computational resource allocation as a MINLP. To tackle this complex problem effectively, we propose a modified branch-and-bound method (Section \ref{sec:ProblemSolution}).
\item Experiments show that our proposed methods can significantly improve the efficiency of SPARQL query processing in edge-cloud computing (Section \ref{sec:Experiments}). 
\end{itemize} 

\section{Background}\label{sec:Background}
\subsection{RDF and SPARQL}
An RDF dataset can be represented as a graph where subjects and objects are vertices and triples are labeled edges.

\begin{definition}\label{def:graph} \textbf{(RDF Graph)}
An RDF graph is denoted as $G=\{V,E,L,f\}$, where $V$ is a set of vertices that correspond to all subjects and objects in the RDF data; $E \subseteq V \times V$ is a multiset of directed edges that correspond to all triples in the RDF data; $L$ is a set of edge labels; and $f:E\to L$ is a label mapping, where for each edge $e \in E$, its edge label $f$($e$) is its corresponding property. \qed
\end{definition}

A SPARQL query can similarly be represented as a query graph $Q$. In this study, we focus on BGP queries as they are foundational to SPARQL, and we consider techniques for handling them.

\begin{definition}\label{def:query}\textbf{(SPARQL BGP Query)}
A \emph{SPARQL BGP query} is denoted as $Q=\{V^{Q},E^{Q}, L^{Q},f^{Q}\}$, where $V^{Q} \subseteq V\cup V_{Var}$ is a set of vertices, where $V$ denotes all vertices in the RDF graph $G$, and $V_{Var}$ is a set of variables; $E^{Q} \subseteq V^{Q} \times V^{Q}$ is a multiset of edges in $Q$; $L^{Q}\subseteq L\cup L_{Var}$ is a set of edge labels, where $L$ denotes all properties in the RDF graph $G$, and $L_{Var}$ is a set of variables for properties; and $f^Q:E^{Q}\to L^{Q}$ is a mapping, where each edge $e$ in $E^Q$ either has an edge label $f^Q(e)$ in $L$ (i.e., property) or the edge label is a variable. \qed
\end{definition}

We assume query $Q$ is a weakly connected directed graph; otherwise, each connected component is considered separately. 

\begin{example}\textbf{(RDF and SPARQL)}
Fig. \ref{fig:ExampleGraphWorkload} illustrates an RDF graph from DBpedia \cite{Dataset:DBpedia} and three SPARQL BGP queries from QALD benchmark \cite{Dataset:QALD}. 
We include the SPARQL syntax of the first query to highlight the mapping between its structure and the query graph.
\end{example}

Answering a SPARQL query is equivalent to finding all subgraphs of $G$ that are homomorphic to $Q$. The subgraphs of $G$ homomorphic to $Q$ are called \emph{matches} of $Q$ over $G$.

\begin{definition}\label{def:sparqlmatch}\textbf{(SPARQL BGP Match)}
Consider an RDF graph $G$ and a query graph $Q$ with vertices $\{v_1,...,v_n\}$. A subgraph of $G$ with vertices $\{u_1, ...,u_x\}$ ($x\leq n$) is a \emph{match} of $Q$ if and only if there exists a \emph{function} $\mu$ from $\{v_1,...,v_n\}$ to $\{u_1,...,u_x\}$ such that:
  1) if $v_i$ is not a variable, $\mu(v_i)$ and $v_i$ have the same URI or literal value ($1\leq i \leq n$);
  2) if $v_i$ is a variable, there is no constraint over $\mu(v_i)$ except that $\mu(v_i)\in \{u_1,...,u_x\}$;
  3) if there exists an edge $\overrightarrow{v_iv_j}$ in $Q$, there also exists an edge $\overrightarrow{\mu{(v_i)}\mu{(v_j)}}$ in $G$;
  4) there must exist an \emph{injective function} from edge labels in $f^{Q}(\overrightarrow{v_iv_j})$ to edge labels in $f(\overrightarrow{\mu{(v_i)}\mu{(v_j)}})$. Note that a variable edge label in $f^{Q}$($\overrightarrow{v_iv_j}$) can match any edge label in $f(\overrightarrow{\mu{(v_i)}\mu{(v_j)}})$.
For $Q$, its matches over an RDF graph $G$ is denoted as $MS(Q,G)$. \qed
\end{definition}

\begin{table}[]
\small
\centering
\caption{Notations}
\resizebox{\columnwidth}{!}{%
\begin{tabular}{ll}
\hline
 \textbf{Notation}& \textbf{Description}\\
 \hline
  $\mathcal{N}$ &  Set of $N$ end users  \\
  $\mathcal{K}$ &  Set of $K$ edge servers  \\
  $\mathbb{D}$ &  Query Assignment decision strategy  \\
  $\mathbb{F}$ &  Computing resource allocation strategy  \\
  $p$ &  Pattern corresponding a query in the workload  \\
  $Q_n$ &  Query task of $EU_n$  \\
  $c_n$ &  Workload of query task $Q_n$  \\
  $w_n$ &  Output data of query task $Q_n$  \\
  $e_{n,k}$ &  Query executable ability of $ES_k$ for query task $Q_n$  \\
  $r^{n,c}$ &  Downlink data rate from Cloud to $EU_n$   \\
  $r^{n,k}$ &  Downlink data rate from $ES_k$ to $EU_n$  \\
  $F_k$ &  Computing resource of $ES_k$  \\
  $f_{n,k}$ &  Computing resource that $ES_k$ allocates to query task $Q_n$ \\
  $\mathcal{N}_k$ &  Set of EUs that assign their query tasks to the edge server $k$ \\
  $\mathcal{N}_{es}$ &  Set of EUs that assign their query tasks to the edge servers \\
  $\mathcal{N}_c$ &  Set of EUs that assign their query tasks to the cloud \\
 \hline
  \end{tabular}%
  }
\label{table:notation}
\end{table}
\setlength{\textfloatsep}{0pt}


\subsection{Edge-Cloud System}
Our edge-cloud system consists of $K$ edge servers (ESs) and $N$ end users (EUs), where each EU may be associated with multiple ESs. The purpose of edge servers is to efficiently process and serve data to the end users, reducing latency and enhancing user experience.
The cloud stores the complete RDF graph $G$, while each ES $k$ stores a subgraph $G_k \subset G$. This means that the information available to each edge server is limited to its stored subgraph. 

EUs submit queries to the system, which then determines whether a query  can be processed by the associated ESs or must be handled by the cloud. The decision is based on the nature of the query and the subgraphs stored in the edge servers. 
If no ES can process a query, it is assigned to the cloud; otherwise, a scheduling decision is made (Section \ref{sec:ProblemSolution}) on whether to execute it at an ES or the cloud.

In this paper, we assume that EUs do not store any subgraphs. This assumption is based on the fact that providing a functional SPARQL query service requires access to a great portion of the RDF graph, which usually surpasses the capabilities of an EU. Although caching the results of specific queries at EUs is a possibility, such caching strategies are orthogonal to our focus on optimizing query task scheduling through offloading, and are not considered in this paper. Thus, EUs rely on edge servers or the cloud to process their queries. 

\section{Problem Formulation}\label{sec:Problem}
In this section, we first clarify the optimization objective and provide a detailed description of the system model. The system model consists of four key components: the data model, query model, communication model, and computation model. Based on these components, we formulate the joint optimization problem of query assignment and computational resource allocation within the edge-cloud environment.

\subsection{Optimization Objective}
In edge-cloud system, given a set $\mathcal{N}$ of end users and a set $\mathcal{K}$ of edge servers, our goal is to allocate resources and assign each query to the most suitable server to enhance overall system performance. The primary objective is to minimize the total response time of queries submitted by end users. 

Each query has two possible execution options: processing the query at the cloud or at one of the associated edge servers.  
Let $\mathbb{D}=\{D_{n,1},...,D_{n,k}\}, \forall n\in \mathcal{N}$ denote the K-dimensional assignment decision vector of $EU_n$ to $ES_k$, where $D_{n,k}\in \{0,1\}$. Consequently, we have:

{\begin{equation}
    D_{n,k} = \left\{
    \begin{array}{rcl}
       1  & \text{if } EU_n \text{ assigns to the } ES_k,\forall k \in \mathcal{K},\\
       0  &  \text{otherwise.}
    \end{array}
    \right.
\end{equation}}

Because each edge server stores only a subset of the complete RDF graph, an end user associated with a particular edge server cannot always have their query processed there. Thus, in addition to the K-dimensional assignment decision vectors, we define $\mathbb{E} = \{e_{n,1}, \dots, e_{n,k}\}, \forall n \in \mathcal{N}$, as the query executable vector. For query $Q_n$, $e_{n,k}\in \{0,1\}$ denotes whether the query submitted by $EU_n$ can be handled by $ES_k$. Specifically, $e_{n,k} = 1$ if the query can be executed on $ES_k$:

{
{\begin{equation} \label{eq:qev}
e_{n,k}=\left\{  
    \begin{array}{rcl}
       1  & \text{if } ES_k \text{ can  handle   query }Q_n \text{ in } EU_n, \\
       0  &  \text{otherwise.}
    \end{array} 
\right.  
\end{equation} }}

Based on the above definitions, each query $Q$ must be processed either by the cloud or by one of its associated edge servers. Specifically, if $\sum_{k\in K}D_{n,k}\times e_{n,k} = 0$, the query is handled by the cloud. Additionally, for all $n \in \mathcal{N}$, it holds that $\sum_{k\in \mathcal{K}}D_{n,k} \times e_{n,k} \le 1$.

{Let $O_{total}$ denote the total execution cost for all queries.
Each query $Q_n$ is executed either on an edge server $k$, with cost $O^{n,k}_e$, or on the cloud, with cost $Q^n_c$. Based on the definitions of $\mathbb{D}$ and $\mathbb{E}$, the total cost is expressed as:}

{{\begin{equation}\label{eq:totalCostEq}
   O_{total}=\sum_{n\in \mathcal{N}} \sum_{k\in \mathcal{K}} D_{n,k}\times e_{n,k}\times  O^{n,k}_e + \sum_{n\in \mathcal{N}} (1-\sum_{k\in \mathcal{K}} D_{n,k}\times e_{n,k})\times Q^n_c.
\end{equation}
}}

\subsection{System Model}
{The execution cost model characterizes the computation and communication overhead at the edges and cloud, which are the dominant contributors to latency.
The model follows a commonly adopted modeling paradigm in mobile edge computing studies \cite{TMC21:MultiTaskLearning,8454442,9340353,8314696}.
To highlight the key factors and avoid unnecessary model complexity, the model does not explicitly consider system factors. Although these factors may affect absolute latency, they typically do not alter the relative performance trends under the same system configuration. }

Since assignment and allocation decisions are determined by the underlying system characteristics, we define the system model as follows.

\subsubsection{Data Model}\label{sec:DataModel}
In our system, the data model comprises not only the global RDF graph stored on the cloud but also a set of subgraphs stored on edge servers. After discussing the optimization objective, we recognize that data localization is one of the key factors in solving this issue. As mentioned before, many real frequently submitted queries share the same graph structures. We use these recurring graph patterns as the basis for defining the subgraphs stored in the edge servers, which can be used for supporting efficient data localization. 

Generally, a pattern \emph{generalizes} a specific query and is an abstract query template that captures only the structural shape of the query, while a query is an \emph{instance} of a pattern.

\begin{definition} \label{def:patterndef} \textbf{(Pattern)} A query $Q$ in the workload corresponds to a \emph{pattern} $p$, where $p$ is generated by replacing all constants (literal and URIs) at subjects and objects with variables in $Q$. \qed
\end{definition}

\nop{
A pattern-induced subgraph consist of subgraphs extracted from graph data that are homomorphic to a specific pattern. In other words, it captures all data fragments within the graph that match the structural shape of the pattern, thereby serving as storage units for edge servers. Therefore, we define \emph{pattern-induced subgraph} as follows:}

{To facilitate efficient data localization, we introduce \emph{pattern-induced subgraphs}. The standard semantics of SPARQL BGP matching is based on subgraph homomorphism; hence, using homomorphism for constructing pattern-induced subgraphs is necessary to preserve result completeness.}
In other words, it captures all data fragments within the graph that match the structural shape of the pattern, thereby serving as storage units for edge servers. Therefore, we define \emph{pattern-induced subgraph} as follows:

\begin{definition} \label{def:PatternInducedSubgraph} \textbf{(Pattern-Induced Subgraph)}
Given a set of patterns $P$, the \emph{pattern-induced subgraph} of $G$ by $P$, denoted as $G[P]$, is the subgraph consisting of both the vertices and edges involved in at least one match of a pattern $p\in P$ in $G$. A match $\mu \in MS(p)$ is a homomorphism from pattern $p$ to a subgraph of $G$,  where $V(\mu)$ and $E(\mu)$ denote the vertices and edges, respectively, in this match. Thus: $G[P] = ( \bigcup\nolimits_{p \in P} {\bigcup\nolimits_{\mu  \in MS(p)} {V(\mu )} }, \bigcup\nolimits_{p \in P} {\bigcup\nolimits_{\mu  \in MS(p)} {E(\mu )} }) $, where $MS(p)$ is the set of all matches of pattern $p$ in $G$. \qed
\end{definition}


It is worth noting that pattern-induced subgraphs may overlap, which means that subgraphs matched by different patterns in the data graph may share common vertices and edges. Using the pattern-induced subgraphs stored on an edge server, data localization becomes simpler and more efficient. 

{
However, the storage capacity of edge servers is inherently limited, making it impractical to deploy all pattern-induced subgraphs. 
To avoid inefficient storage use, edge servers adopt a storage-aware selection mechanism when deploying pattern-induced subgraphs. Specifically, the pattern selection process is abstracted as a knapsack problem, where the access frequency of a pattern reflects its benefit, and the size of its subgraph represents the storage cost. The system employs a lightweight greedy heuristic to prioritize patterns that offer high benefit under limited storage budgets. The mechanism only influences the execution location and performance.
}

When end users submit queries, the system first transmits the query to the cloud for centralized scheduling. The critical step in scheduling is to determine whether a given query can be executed directly on a specific edge server, based on the patterns deployed on that server.
In contrast to the homomorphism used for constructing pattern-induced subgraphs, query executability is determined based on graph isomorphism. A query is executable if its query graph is isomorphic to one of the patterns whose induced subgraphs are stored on the edge server. The distinction between homomorphism in construction and isomorphism for execution is essential, because two graphs that are each homomorphic to a third graph need not be homomorphic to one another \cite{DiscreteMathematics:HomomorphismsProperties}. Thus, even if a query is homomorphic to a pattern, it does not follow that all matches of the query are contained within the matches of that pattern. 
We further illustrate this point using the example in Fig. \ref{fig:ExampleHomomorphismIsomorphism}. The graph $K_2$ is homomorphic to $K_3$, yet the set of subgraphs induced by all matches of $K_3$ does not necessarily include all matches of $K_2$.

\begin{figure}[]
\centering
\subfigure[{K2}]{%
        \centering
		\includegraphics[width=0.45\columnwidth]{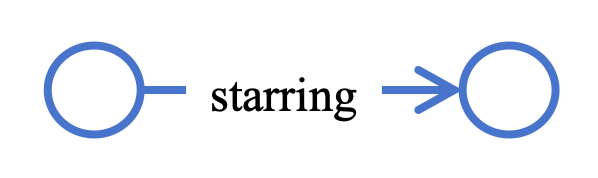}
       \label{fig:K2}%
       }
   \subfigure[K3]{%
        \centering
		\includegraphics[width=0.45\columnwidth]{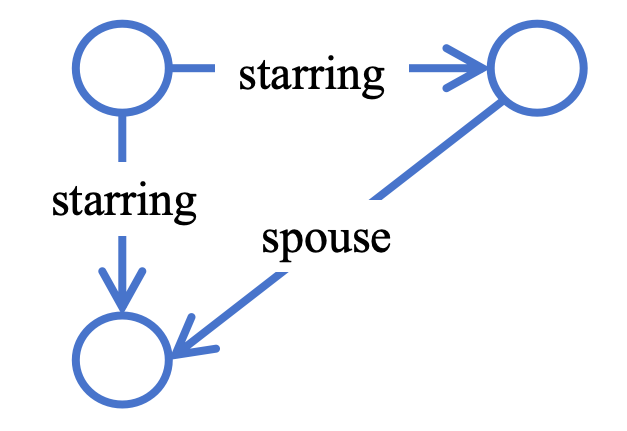}
       \label{fig:K3}%
       }
\caption{Example for homomorphism and isomorphism}%
 \label{fig:ExampleHomomorphismIsomorphism}
\end{figure}

This graph isomorphism check determines the values of the query executable vector $\mathbb{E}$ in Eq. (\ref{eq:qev}). Formally, for a query $Q_n$ and an $ES_k$, we set $e_{n,k}=1$ if $Q_n$ is isomorphic to any pattern deployed on $ES_k$; otherwise, $e_{n,k}=0$.
This binary matrix $\mathbb{E}$ encodes the compatibility and serves as the foundation for subsequent scheduling.

To efficiently construct the query executable vector $\mathbb{E}$, we design a lightweight indexing mechanism for the patterns stored on edge servers.
Specifically, each pattern is first transformed into its minimal DFS code, which serves as a canonical sequence representation of the graph structure \cite{SIGMOD04:gIndex}. 
The minimum DFS code selects the lexicographically smallest DFS traversal sequence among all valid ones, ensuring that two patterns share the same code if and only if they are isomorphic. Although computing the minimum DFS code can be expensive in theory, the overhead is negligible in our setting since patterns are small (typically containing fewer than 10 triples), which limits the search space and results in negligible overhead in practice.
The resulting canonical DFS codes are then hashed and indexed using a hash table, When a query arrives, it is encoded in the same manner, allowing the system to quickly identify matching patterns. This indexing mechanism avoids expensive subgraph matching at runtime, thereby reducing computation overhead and query latency.

\begin{example}For the workload in Fig. \ref{fig:FirstExampleQuery}, the patterns corresponding to the three queries are shown in Fig. \ref{fig:ExamplePatterns}. 
In this example, $G_1$ is a subgraph induced by $p_1$ and $p_2$ in Fig. \ref{fig:ExampleGraph} involving the blue and gray vertices, while $G_2$ is a subgraph induced by $p_2$ and $p_3$ involving the gray and yellow vertices. In the example edge-cloud system, $G_1$ and $G_2$ are stored at edge servers $ES_1$ and $ES_2$, respectively.
\end{example}

{To handle evolving query workloads, we introduce a dynamic update mechanism. 
Specifically, the system continuously monitors the access frequency of each query pattern. When certain patterns appear frequently in the cloud but are not present on the edge, the system dynamically adds the corresponding pattern-induced subgraphs to the edge servers. To control storage overhead, the system also evicts inactive or outdated patterns based on frequency statistics. To ensure query processing efficiency, this update process is implemented as an asynchronous background task, decoupled from the query response flow, and thus does not affect online query latency.}

\begin{figure}
\centering
        \centering
		\includegraphics[width=\columnwidth]{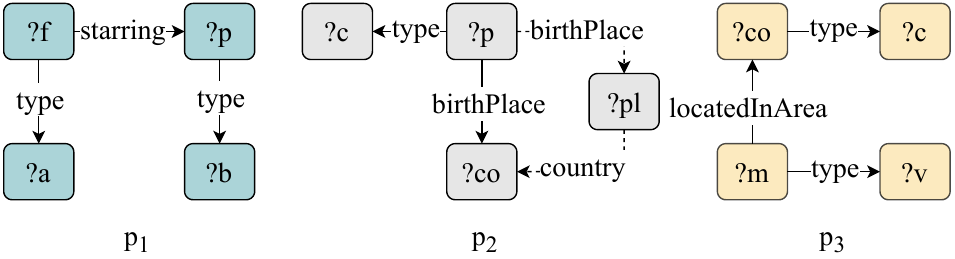}
{\caption{Example Patterns}}%
 \label{fig:ExamplePatterns}
\end{figure}



\subsubsection{Query Model}\label{queryModel}
We assume that each EU submits SPARQL queries one-at-a time. 
A query is characterized as a two-tuple of parameters, i.e., $Q_n=(c_n, w_n)$, where $c_n$ [cycles] denotes the amount of computation to execute $Q_n$ (i.e., the total number of CPU cycles required to process it), and $w_n$ [bits] is the result size of $Q_n$ in number of bits.

The cost estimation of a SPARQL query is a well-studied problem. In this paper, we directly adopt existing methods \cite{10.1145/1367497.1367578,10.1145/1559845.1559911} to estimate the amount of computation and result size. These techniques are orthogonal to our focus, and our formulation can accommodate alternative cost functions.


As discussed in Section \ref{sec:DataModel}, we use a hash table to efficiently locate patterns that are homomorphic to a given query. Then, $e_{n,k} = 1$ if and only if the pattern-induced subgraph of the query is stored on $ES_k$. 

\begin{example}
In the edge-cloud system shown in Fig. \ref{fig:system_overview}, we assume that a query homomorphic to $p_1$ is submitted by end user $EU_3$. Although $EU_3$ is associated with two edge servers, $ES_1$ and $ES_2$, after using the pattern-induced subgraphs stored in $ES_1$ and $ES_2$ for data localization, we can determine that the query can be processed by $ES_1$ but not by $ES_2$. Therefore, $e_{3,1} = 1$ and $e_{3,2} = 0$.
\end{example} 

\subsubsection{Communication Model}

In the model, edge servers can communicate with the cloud, but not among themselves. Therefore, each submitted query is either executed on an edge server or assigned to the cloud for precessing, and the results are subsequently returned to the end user.

The transmission rate for transmitting the result of query $Q_n$ from the cloud to EUs is assumed to be fixed \cite{8318578,9861697}, and is denoted as $r^{n,c}$.
{When the query is executed at a ES, the transmission rate to the EU depends on the wireless access link.
We consider the MEC architecture based on orthogonal frequency division multiple-access (OFDMA) \cite{DBLP:conf/icc/FengZDCY18,book:MobileBroadband}, which avoids mutual interference among users by allocating orthogonal radio resources.
Each EU is allocated an identical bandwidth $B$. The transmission power of edge server $k$ and the channel gain between edge server $k$ and end user $n$ are denoted by $tp_k$ and $h_{k,n}$, respectively. Accordingly, the transmission rate from $ES_k$ to $EU_n$ is given by:}

{{\begin{equation}
   r^{n,k}=B\times \log_2(1+\frac{tp_k\times h_{k,n}}{\sigma^2}),
\end{equation}}}
where $\sigma^2$ is the background noise.

\subsubsection{Computation Model}


\paragraph{Edge Computing}
The computational capability of each edge server is inherently limited. For $ES_k$, its total computational resource is denoted as $F_k$, expressed in terms of the number of CPU cycles per second. Let $F_n = \{f_{n,1}, \dots, f_{n,K}\}$ denote the $K$-dimensional vector of computational resource (in CPU cycles per second) allocated to $EU_n$ by $ES_k$, representing the distribution of resources across edge servers.

The allocation result, denoted as $\mathbb{F} = \{F_{1}, \dots, F_{N}\}$, represents the allocation of computational resources from edge servers to the various query tasks submitted by end users. In order to ensure the validity of the allocation, it must satisfy the constraints $0\le f_{n,k}\le F_k$ and $\sum_{n=1}^{N} f_{n,k}\le F_k$.

The cost associated with processing a query can be broken down into three components: 
\begin{enumerate}[left=0pt]
\item \textbf{Query Uploading.} The cost of uploading queries is negligible and is excluded from the total cost calculation.
\item \textbf{Query Execution.} The cost of executing queries is defined as $\frac{c_n}{f_{n,k}}$, representing the ratio of the computational requirement $c_n$ for a query to the allocated computational resources $f_{n,k}$ provided by $ES_k$ to $EU_n$.
\item \textbf{Result Forwarding.} The cost of forwarding results to end users is defined as $\frac{w_n}{r^{n,k}}$, where $w_n$ denotes the size of the query result and $r^{n,k}$ indicates the achieved data rate of the wireless link between $EU_n$ and $ES_k$. 
\end{enumerate}

Hence, the total cost of executing query $Q_n$ at the $ES_k$, denoted as $O_e^{n,k}$, is expressed as:
    $O_e^{n,k}=\frac{c_n}{f_{n,k}} + \frac{w_n}{r^{n,k}}$.
By integrating the execution time and result forwarding time, the system can efficiently determine the total cost of processing a query on the edge servers. This information is crucial for optimizing resource allocation and selecting the most suitable computing node for each query, ultimately improving system performance and user experience.

\paragraph{Cloud Computing}
Recall that the cloud has unlimited computational capability in the model, and the computation delay is negligible compared to that of transmission. As in edge computing, the cost of uploading queries to the cloud is also negligible.
Thus, the cost of executing query $Q_n$ in the cloud, $O_c^n$, is dominated by the time required to forward results to the user, given by $O_c^n = \frac{w_n}{r^{n,c}}$. Where $w_n$ is the size of the query result and $r^{n,c}$ indicates the achieved data rate of the wireless link between $EU_n$ and the cloud.


By comparing this cost with that of processing the query on the edge servers, the system can effectively determine the most efficient computation node for each query, leading to optimized performance and an enhanced user experience.

\subsection{Problem Formulation}
Based on the above discussions, given a set $\mathcal{N}$ of end users, a set $\mathcal{K}$ of edge servers and a set of queries, the total cost of all queries defined in Eq. (\ref{eq:totalCostEq}) can be further expressed as:

{\begin{equation} \label{eqn:O_total}
\begin{aligned}
    O_{total}(\mathbb{D}, \mathbb{F})=\sum_{n\in \mathcal{N}}(\sum_{k\in \mathcal{K}}D_{n,k}\times e_{n,k} \times (\frac{c_n}{f_{n,k}} + \frac{w_n}{r^{n,k}}) \\
    + (1-\sum_{k\in \mathcal{K}} D_{n,k}\times e_{n,k})\times \frac{w_n}{r^{n,c}})).
\end{aligned}
\end{equation}}

Having provided the textual definition of the optimization objective earlier, we now formalize it using the following mathematical representation.

\begin{definition}\label{def:problemdef}\textbf{(Joint Formulation  of  Query Assignment and Computational Resource Allocation)} Given the sets of EUs and ESs in the system, the joint optimization problem is formulated as:

{\begin{equation}
\begin{aligned}
    & \min_{\mathbb{D},\mathbb{F}} O_{total} (\mathbb{D}, \mathbb{F}) \\
    \textbf{s.t.} \quad
    & C_1:  D_{n,k} \in \{0,1\} \\
    & C_2: \sum_{k\in \mathcal{K}}D_{n,k}\times e_{n,k} \le 1 \\
    & C_3:  0 \le f_{n,k}\le F_k  \\
    & C_4: \sum_{n\in \mathcal{N}} f_{n,k}\le F_k\\ \label{opt:p1}
\end{aligned}
\end{equation}}

with $O_{total}$ defined in Eq. \eqref{eqn:O_total}.
\end{definition}

The solution to this problem involves finding the optimal assignment decision vector $\mathbb{D}$ and the optimal allocation vector of computational resources $\mathbb{F}$. The constraint $C_1$ shows that $EU_n$ can only choose to execute query $Q_n$ in the cloud or at an edge server $ES_k$. $C_2$ specifies that a query needs to be handled by only one edge server or by the cloud. $C_3$ and $C_4$ ensure that the computational resources allocated to all end users whose queries are assigned to $ES_k$ should not exceed the resources at $ES_k$.

\begin{table}[]
\centering
\caption{Example Solution}
\begin{tabular}{c c @{\hspace{1cm}} c}
$\mathbb{D}=$ & 
$\begin{bmatrix}
1 & 0 \\
1 & 0 \\
0 & 1 \\
0 & 0 \\
0 & 1
\end{bmatrix}$ 
& 
$\mathbb{F}=$ \quad 
$\begin{bmatrix}
44.4 & 0 \\
35.6 & 0 \\
0 & 42.5 \\
0 & 0 \\
0 & 37.5
\end{bmatrix}$
\end{tabular}
 \label{table:example_solution}
\end{table}

\begin{figure}[]
    \centering
      \includegraphics[width=0.85\columnwidth]{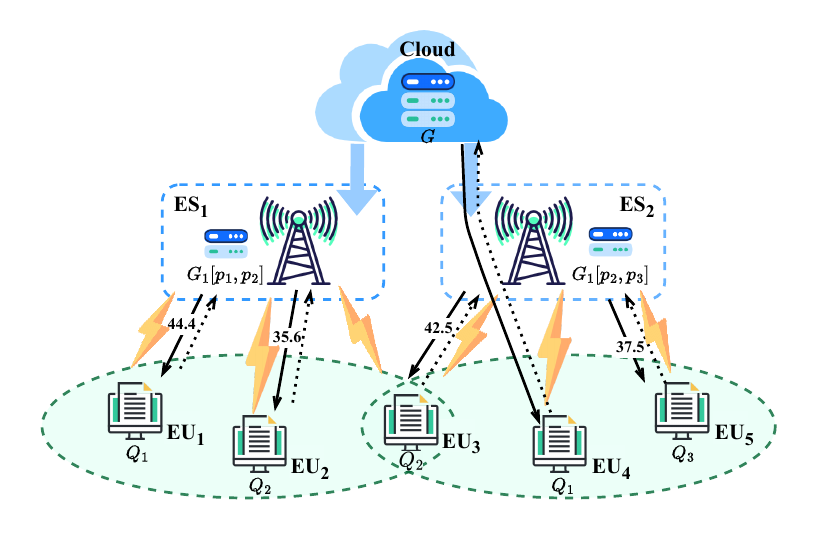}
\caption{Example Solution}%
 \label{fig:example_solution}
\end{figure}

\begin{example}
Table \ref{table:example_solution} and Fig. \ref{fig:example_solution} give a possible solution for the edge-cloud computing system. Here, we assume that the computational resources $F_1$ and $F_2$ of edge servers $ES_1$ and $ES_2$ are both $80$. Each EU submits a query, and the query shown beside each EU represents the one submitted by that user. Queries $Q_1$, $Q_2$, and $Q_3$ are part of the example workload shown in the Fig. \ref{fig:FirstExampleQuery}, which are homomorphic to the patterns in Fig. \ref{fig:ExamplePatterns}.

In the example, edge server $ES_1$ processes queries from $EU_1$ and $EU_2$, while $ES_2$ handles queries from $EU_3$ and $EU_5$. The remaining query from $EU_4$ is assigned to the cloud.
Hence, the fourth row of both $\mathbb{D}_{n,k}$ and $\mathbb{F}_{n,k}$ is zeros.
$ES_1$ allocates computational resources of $44.4$ and $35.6$ to these queries, with $\mathbb{D}_{1,1}=\mathbb{D}_{2,1}=1$, $\mathbb{F}_{1,1}=44.4$, and $\mathbb{F}_{2,1}=35.6$.
\end{example}




\section{Problem Transformation And Solution}\label{sec:ProblemSolution}
In this section, we first decompose the problem into two subproblems: the query assignment decision and the computational resource allocation problems. Next, we outline our approach to solve the optimization problem by addressing the computational resource allocation problem and using its solutions to derive the query assignment decision problem.

\subsection{Problem Analysis}
The optimization problem of joint query assignment decision and computational resource allocation is a mixed-integer nonlinear programming problem (MINLP), which is a NP-hard problem \cite{DBLP:journals/actanum/BelottiKLLLM13}. Finding the optimal solution usually requires exponential time complexity. Intuitively, the optimization problem can be solved by exploring all combinations of the assignment decision vector ($\mathbb{D}$) and computational resource allocation ($\mathbb{F}$). However, since $\mathbb{D}$ is a binary decision vector, the search space grows exponentially with its dimension, making the problem intractable for large-scale instances.

The structure of the objective function and constraints (i.e., constraints $C_1-C_4$ in \eqref{opt:p1}) allows the original problem to be decomposed into two subproblems. Specifically, by fixing the binary assignment decision variable $\mathbb{D}$, the decomposition can be carried out using the Tammer method \cite{Tammer1987TheAO}. The original optimization problem can thus be reformulated as:

{\begin{equation}
\begin{aligned}
    &\min_{\mathbb{D}} (\min_{\mathbb{F}} O_{total} (\mathbb{D}, \mathbb{F})) \label{opt:equivalent_p1} \\
    &\textbf{s.t. } \text{constraints $C_1$-$C_4$ in Eq. \eqref{opt:p1}}
\end{aligned}
\end{equation}}

Note that the constraints $C_1$ and $C_2$ on the assignment decision, $\mathbb{D}$, and the constraints $C_3$ and $C_4$ on the computational resource allocation, $\mathbb{F}$, are decoupled from each other. The equivalent problem in \eqref{opt:equivalent_p1} can be decomposed into two subproblems: the query assignment decision (QAD) subproblem, which minimizes cost, and the computational resource allocation (CRA) subproblem for a fixed assignment:

{\begin{equation}
\begin{aligned}
    & \min_{\mathbb{D}} O_{total}^* (\mathbb{D}) \label{opt:qod}  \\
    & \textbf{s.t. } \text{constraints $C_1$, $C_2$ in Eq. \eqref{opt:p1}}
\end{aligned}
\end{equation}}

where $O_{total}^* (\mathbb{D})$ denotes the optimal value of the computational resource allocation problem, defined as follows:

{\begin{equation}
\begin{aligned}
    & O_{total}^* (\mathbb{D}) = \min_{\mathbb{F}} O_{total} (\mathbb{D}, \mathbb{F}) \label{opt:cra} \\ & \textbf{s.t. }  \text{constraints $C_3$, $C_4$ in Eq. \eqref{opt:p1}}
\end{aligned}
\end{equation}}

Note that the decomposition from the original problem \eqref{opt:p1} to subproblems \eqref{opt:qod} and \eqref{opt:cra} does not change the optimality of the solution \cite{TVT19:JointTaskOffloading}. In  this paper, we will present our solutions to both the CRA problem and the QAD problem so as to finally obtain the solution to the original problem.


Firstly, we denote the set of end users (EUs) that assign the query to the $ES_k$ as $\mathcal{N}_k = \{ n \in N| D_{n,k} \times e_{n,k} = 1\}$. Consequently, $\mathcal{N}_{es} = \bigcup_{k \in \mathcal{K}} \mathcal{N}_k$ represents the set of EUs assigning the query to edge server, while $\mathcal{N}_{c} = \mathcal{N} - \mathcal{N}_{es}$ represents the set of EUs assigning the query to the cloud.

Using the definitions of the edge and cloud execution costs, the objective function in Eq. \eqref{eqn:O_total} can be rewritten as:

{\begin{equation} \label{eqn:O_total_rewritten}
    O_{total} = \sum_{k\in\mathcal{K}}\sum_{n\in\mathcal{N}_k}\frac{c_n}{f_{n, k}} + \sum_{k\in\mathcal{K}}\sum_{n\in\mathcal{N}_k}\frac{w_n}{r^{n, k}}+ \sum_{n\in\mathcal{N}_{c}}\frac{w_n}{r^{n, c}}
\end{equation}}

\subsection{Computational Resource Allocation (CRA)}
We observe that the second and the third term on the right hand side of \eqref{eqn:O_total_rewritten} is constant for a particular assignment decision, while the first term can be seen as the total computational overhead of all assigned queries. Hence, given a feasible query assignment decision $\mathbb{D}$ that satisfies constraints $C_1$ and $C_2$ in Eq. \eqref{opt:p1}, we can recast \eqref{opt:cra} as the problem of minimizing the total computational overhead:

{\begin{equation}
\begin{aligned}
    & \min_{\mathbb{F}} O_{total\_calc}(\mathbb{D},\mathbb{F}) = \sum_{k\in\mathcal{K}}\sum_{n\in\mathcal{N}_k}\frac{c_n}{f_{n,k}} \label{opt:cra_calc}  \\ & \textbf{s.t. }  \text{constraints $C_3$, $C_4$ in Eq. \eqref{opt:p1}}
\end{aligned}
\end{equation}}

Notice that the constraints $C_3$ and $C_4$ in Eq. \eqref{opt:p1} are convex, and the Hessian matrix of the objective function in Eq. \eqref{opt:cra_calc} is diagonal with strictly positive elements, making it positive-definite \cite{TVT19:JointTaskOffloading}. Thus, the optimization problem in Eq. \eqref{opt:cra_calc} is convex and can be solved using the Karush-Kuhn-Tucker (KKT) conditions. 


By equating the gradient of the Lagrangian to zero, we derive the optimal computational resource allocation $f_{n,k}^*$ and its corresponding optimal objective function $O_{total\_calc}(\mathbb{D}, \mathbb{F}^*)$ for solving Eq. \eqref{opt:cra_calc}, which are calculated as:

{\begin{equation}
    f_{n,k}^*=\frac{F_k\sqrt{c_n}}{\sum_{n\in\mathcal{N}_k} \sqrt{c_n}}, \forall k \in\mathcal{K}, n \in\mathcal{N}_k, \label{eqn:f_nk}
\end{equation}
\begin{equation}
    O_{total\_calc}(\mathbb{D}, \mathbb{F}^*)=\sum_{k\in\mathcal{K}}\frac{1}{F_k}(\sum_{n\in\mathcal{N}_k}\sqrt{c_n})^2. \label{eqn:o_calc}
\end{equation}}

\begin{example}
Given the feasible query assignment decision $\mathbb{D}$ in Table \ref{table:example_solution}, the sets of EUs assigned to each edge server are determined: $N_1=\{EU_1, EU_2\}$ and $N_2=\{EU_3, EU_5\}$. With these allocations, the optimal computing resource allocation for each $EU$ is computed using \eqref{eqn:f_nk}.
\end{example}


\subsection{Query Assignment Decision (QAD)}
In the previous section, for a given query assignment decision $\mathbb{D}$, we obtained the solution for the computing resource allocation. In particular, according to Eqs. \eqref{opt:cra}, \eqref{eqn:O_total_rewritten}, \eqref{opt:cra_calc} and \eqref{eqn:o_calc}, we have:

{\begin{equation} \label{eqn:o_total_1}
    O_{total}^* = O_{total\_calc} +\sum_{k\in\mathcal{K}}\sum_{n\in\mathcal{N}_k}\frac{w_n}{r^{n,k}}+ \sum_{n\in\mathcal{N}_{c}}\frac{w_n}{r^{n,c}}
\end{equation}}

where $O_{total\_calc}$ can be calculated using the closed-form expression in Eq. \eqref{eqn:o_calc}. Now, using Eq. \eqref{eqn:o_total_1}, we can rewrite the QAD problem in Eq. \eqref{opt:qod} as:

{\begin{equation}
\begin{aligned}
    & \min_{\mathbb{D}} \sum_{k\in\mathcal{K}}\frac{1}{F_k}(\sum_{n\in\mathcal{N}_k}\sqrt{c_n})^2 +\sum_{k\in\mathcal{K}}\sum_{n\in\mathcal{N}_k}\frac{w_n}{r^{n,k}}+ \sum_{n\in\mathcal{N}_{c}}\frac{w_n}{r^{n,c}} \label{opt:qod_rewritten}\\
    \textbf{s.t.}  \
    & C_1: D_{n,k} \in \{0,1\};   \hspace{1em}
    C_2: \sum_{k\in \mathcal{K}}D_{n,k}\times e_{n,k} \le 1.
\end{aligned}
\end{equation}}

\subsection{Branch-and-Bound Algorithm}
In this subsection, we extend techniques in \cite{DBLP:journals/access/ZhaoTQN17} to introduce a branch-and-bound  algorithm to achieve the globally optimal solution for the problem defined in Eq. \eqref{opt:qod_rewritten}. 

We first build a search tree, where each node corresponds to a partial solution. This tree is generated using a depth-first strategy, with the root node representing the problem in Eq. \eqref{opt:qod_rewritten}. To effectively prune branches and search for an optimal strategy, it's essential to calculate upper ($\overline{O_{total}}$) and lower ($\underline{O_{total}}$) bounds for each partial solution. We select the unpruned leaf node with the maximal depth (i.e., the farthest from the root) and the lowest upper bound to continue branching, until all queries are considered.

\subsubsection{Bounding}
For a node corresponding to the problem in Eq. \eqref{opt:qod_rewritten}, we describe how to calculate its upper and lower bounds. First, we relax the integer variable (i.e. $\mathbb{D}$) to a continuous one. By relaxing the constraints, the problem is simplified as follows, referred to as Problem R-QAD:

{\begin{equation}
\begin{aligned}
    & \min_{\mathbb{D}} \sum_{k\in\mathcal{K}}\frac{1}{F_k}(\sum_{n\in\mathcal{N}_k}\sqrt{c_n})^2 +\sum_{k\in\mathcal{K}}\sum_{n\in\mathcal{N}_k}\frac{w_n}{r^{n,k}}+ \sum_{n\in\mathcal{N}_{c}}\frac{w_n}{r^{n,c}} \\
    \textbf{s.t.}  \
    & C_1: D_{n,k} \in [0,1];  \hspace{1em}  C_2: \sum_{k\in \mathcal{K}}D_{n,k}\times e_{n,k} \le 1.
\end{aligned}\label{opt:qod_lower}
\end{equation}}

The optimal of Eq. \eqref{opt:qod_lower}, $\underline{O_{total}}$, is a \emph{lower bound} of Eq. \eqref{opt:qod_rewritten}. Thus, we define $\underline{\mathbb{D}}$ as the optimal solution to this problem.

\begin{theorem}\label{theorem:LowerConvexOptimization}
Problem R-QAD is a convex optimization problem.
\end{theorem}

\begin{proof}
{We examine the convexity of the objective function and constraints. The objective consists of three parts. the first term, $\sum_{k\in\mathcal{K}}\frac{1}{F_k}(\sum_{n\in\mathcal{N}}D_{n,k}\times e_{n,k}\sqrt{c_n})^2$, has a Hessian matrix with elements $\frac{2e_{n,k}e_{m,k}\sqrt{c_nc_m}}{F_k}$}. Since these values are non-negative, the matrix is positive semi-definite, indicating this part is a convex function. The second term, $\sum_{k\in\mathcal{K}}\sum_{n\in\mathcal{N}}D_{n,k}\times e_{n,k}\frac{w_n}{r^{n,k}}$, is linear, and the third term, $\sum_{k\in\mathcal{K}}\sum_{n\in\mathcal{N}}(1-D_{n,k}\times e_{n,k})\frac{w_n}{r^{n,c}}$, is affine. Since linear and affine functions are convex, the entire objective function is convex.
Meanwhile, $C_1$ is an interval constraint and $C_2$ is a linear inequality constraint; both define convex feasible sets. Thus, Problem R-QAD is a convex optimization problem.
\end{proof}

Given a specified $\underline{\mathbb{D}}$, we apply the following formula to obtain the specific value for the $\overline{\mathbb{D}}$:

{\begin{equation}
    \overline{\mathbb{D}} = \lbrack \ \overline{\mathbb{D}_{n,k}} \mid \ \overline{\mathbb{D}_{n,k}} = \left\{
    \begin{aligned}
        1,  \ \ \underline{\mathbb{D}_{n,k}} \geq 0.5 \\
        0,  \ \  \underline{\mathbb{D}_{n,k}} < 0.5 \\
    \end{aligned}
    \right. \ \ 
    \forall \ \underline{\mathbb{D}_{n,k}} \in \underline{\mathbb{D}} \ \rbrack \label{eqn:CalcDUpper}
\end{equation}}

We determine $\mathcal{N}_k(\forall k \in \mathcal{K})$ and $\mathcal{N}_{c}$ based on $\overline{\mathbb{D}}$ and calculate the optimal value, $\overline{O_{total}}$, which is a \emph{upper bound} of problem (\ref{opt:qod_rewritten}):

{\begin{equation}
    \overline{O_{total}} = \sum_{k\in\mathcal{K}}\frac{1}{F_k}(\sum_{n\in\mathcal{N}_k}\sqrt{c_n})^2 +\sum_{k\in\mathcal{K}}\sum_{n\in\mathcal{N}_k}\frac{w_n}{r^{n,k}}+ \sum_{n\in\mathcal{N}_{c}}\frac{w_n}{r^{n,c}} \label{eqn:qod_upper}
\end{equation}}

\begin{figure*}[t]
    \centering
      \includegraphics[width=1.1\columnwidth]{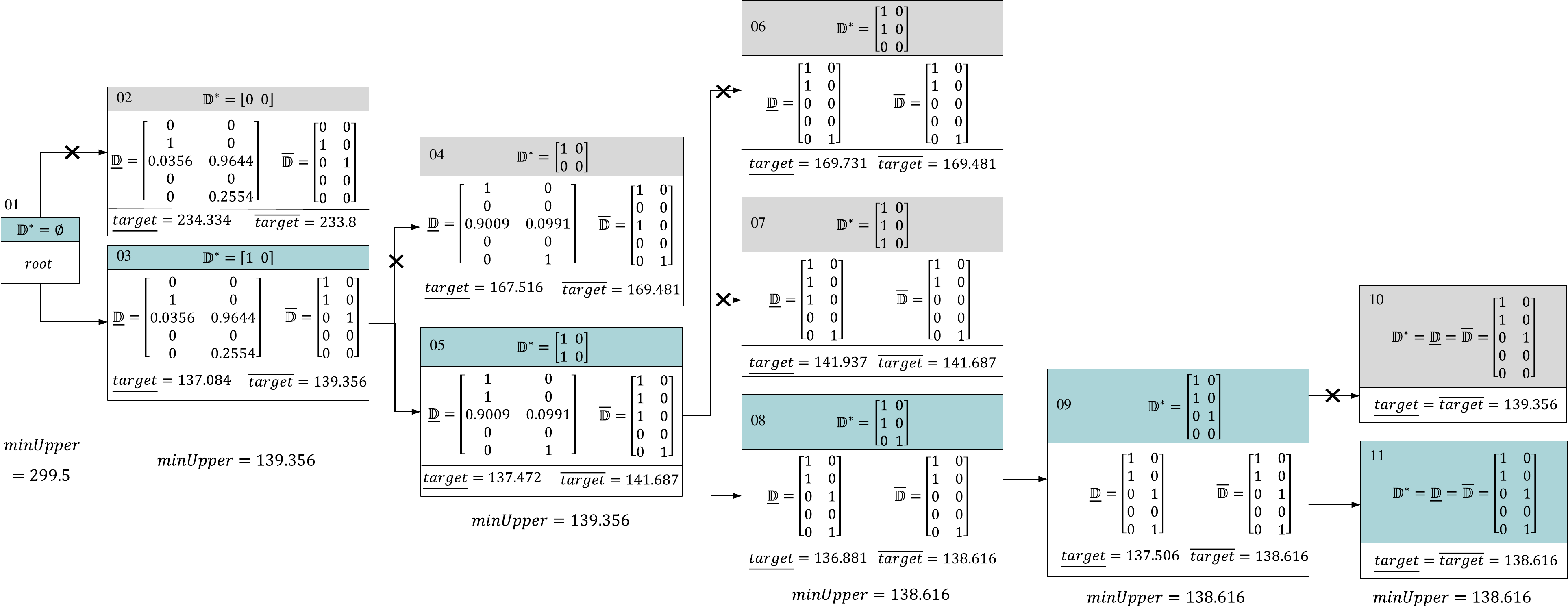}
\caption{Example search tree of the branch-and-bound method}%
 \label{fig:branchandbound}
\end{figure*}

\subsubsection{Branching}
In this context, we consider an iteration within the branching process, denoted as the $i^{th}$ branching. During this iteration, we select the node with maximal depth and the lowest upper bound from the remaining unpruned nodes. Once the node $s_{min}$ is determined, it can be used as the basis for branching into several subproblems. 

We define $\mathcal{N}_d$ as the EUs set for determined execution strategy, and define $O_D^*$ as the cost of $\mathcal{N}_d$:

{\begin{equation}
    O_D^* = \sum_{k\in\mathcal{K}}\frac{1}{F_k}(\sum_{n\in\mathcal{N}_d}\sqrt{c_n})^2 +\sum_{k\in\mathcal{K}}\sum_{n\in\mathcal{N}_d}\frac{w_n}{r^{n,k}}+\sum_{n\in\mathcal{N}_d}\frac{w_n}{r^{n,c}} \label{eqn:DeterminedCost}
\end{equation}}

Therefore, the subproblems are formulated as follows:

{\begin{equation}
\begin{aligned}
    & \min_{\mathbb{D}} O_D^* + \sum_{k\in\mathcal{K}}\frac{1}{F_k}(\sum_{n\in\mathcal{N}_k / \mathcal{N}_d}\sqrt{c_n})^2 +\sum_{k\in\mathcal{K}}\sum_{n\in\mathcal{N}_k / \mathcal{N}_d}\frac{w_n}{r^{n,k}} \\
    & \qquad\qquad +\sum_{n\in\mathcal{N}_{c} / \mathcal{N}_d}\frac{w_n}{r^{n,c}} \label{opt:BranchSubproblem}\\
    \textbf{s.t.}  \
    & C_1: D_{n,k} \in [0,1] \\
    & C_2: \sum_{k\in \mathcal{K}}D_{n,k}\times e_{n,k} \le 1 \\
    & C_3: D^*(|\mathcal{N}_d|) = {D^*(|\mathcal{N}_d-1|, \mathbf{T})}
\end{aligned}
\end{equation}

\noindent where $D^*(|\mathcal{N}_d|)$ denote the determined i-dimensional assignment decision vector of $\mathcal{N}_d$. $\mathbf{T} = (t_j)_{1\times|\mathcal{K}|}$ denote the generated branch ($t_j\in \{0,1\}$). Specifically, when $\sum_{j=1}^{|\mathcal{K}|} t_j=1$, the query is assigned to the edge server for processing. Otherwise, the query is assigned to the cloud.

\subsubsection{Branch-and-Bound Algorithm}
Algorithm 1 outlines the modified branch-and-bound algorithm. It begins by calculating the total cost of the cloud-only assignment (line 3). A priority queue is then initialized to store partial solutions, ranked by their upper bounds(line 5). The algorithm dequeues the partial solution with the smallest upper bound (line 7). If it is a complete solution, it is added to the set of complete solutions (lines 8-10). Otherwise, the partial solution is branched into several larger subproblems, which are processed by calculating their lower and upper bounds (lines 12-18). If a subproblem’s lower bound exceeds the current minimum upper bound, it is pruned (lines 19-21). The process continues until all partial solutions are processed, selecting the complete solution with the smallest total cost.

{\begin{algorithm} \label{alg:branchAndBound}
\caption{Branch-and-Bound Algorithm}

 \KwIn{System parameters $\{|\mathcal{N}|,|\mathcal{K}|,F_k\}$; Query parameters $\{e_{n,k}, c_n, w_n\}$}.
\KwOut{The assignment decision vector $\mathbb{D}$ and the vector of allocated computational resources $\mathbb{F}$}
Initialize the set of complete solutions $CS$;\\
Calculate the cost, $O_{all}^{cloud}$, for cloud-only assignment. \\
$minUpper \gets O_{all}^{cloud}$;\\
Initialize a partial solution and enqueue it into the priority queue $PQ$, where partial solutions are ranked by their upper bounds;\\
\While{$PQ \ne \emptyset$} {
    Dequeue the partial solution $s_{min}$ with the smallest upper bound from $PQ$; \\
    \If{$s_{min}$ is a complete solution} {
        Insert $s_{min}$ into $CS$;\\
        \textbf{continue};\\
    }
    Generate subproblems $S_{sub}$ by branching $s_{min}$ into several larger partial solutions using Eq. \eqref{opt:BranchSubproblem}; \\
    \For{each partial solution $s$ in $S_{sub}$} {
        Calculate $s.\underline{O_{total}}$ and $\underline{\mathbb{D}}$ using Eq. \eqref{opt:qod_lower}; \\
        Calculate $\overline{\mathbb{D}}$ based on $\underline{\mathbb{D}}$ using Eq. \eqref{eqn:CalcDUpper}; \\
        Calculate $s.\overline{O_{total}}$ based on $\overline{\mathbb{D}}$ using Eq. \eqref{eqn:qod_upper};\\
        Enqueue $s$ into $PQ$; \\
        \If{$s.\overline{O_{total}} < minUpper$} {
            $minUpper \gets s.\overline{O_{total}}$;\\
        }
    }
    \For{each partial solution $s^\prime$ in $PQ$} {
        \If{$s^\prime.\underline{O_{total}} > minUpper$} {
            Remove $s^\prime$ from $PQ$; \\
        }
    }  
}
Find the complete solution $s$ with the smallest total cost from $CS$ and set its assignment decision vector as $\mathbb{D}$;\\
Calculate $\mathbb{F}$ based on $\mathbb{D}$ using Eq. \eqref{eqn:f_nk};\\
\Return $\mathbb{D}$, $\mathbb{F}$;
\end{algorithm}}

{In the worst case, the branch-and-bound algorithm needs to explore the entire search space. Its worst-case complexity is determined by the depth of the search tree (the number of users $|\mathcal{N}|$) and the branching factor at each node (the number of server $|\mathcal{K}|+1$), i.e., $O(|\mathcal{K}|^{|\mathcal{N}|})$.
Although the worst-case complexity is exponential, in practice, most invalid branches are pruned early through the use of priority queue sorting and upper- and lower-bound pruning.}
In practice, the number of edge servers is typically limited, and most end users are just associated with only a few edge servers. Thus, we can often efficiently prune a large number of partial solutions during exploration, leading to significant runtime improvements.

\begin{example}
Fig. \ref{fig:branchandbound} illustrates the branch-and-bound search tree for the system in Fig. \ref{fig:example_solution}.
Initializing with the cloud solution, the algorithm first branches on $EU_1$ for $Q_1$, executable on $ES_1$ or in the cloud. This generates two child nodes, and solving these nodes produces an initial upper bound of $minUpper = 139.356$. It pruning node $02$ since its lower bound $234.334$ exceeds $minUpper$. At the second level, branching on $EU_2$ for $Q_2$, generating nodes $04$ and $05$, and prune node $04$. After further branching and pruning, node $11$ represents the optimal solution.
\end{example}

\section{Experiments}\label{sec:Experiments}
\subsection{Setting}
We construct an edge-cloud computing environment using one cloud server and multiple virtual machines representing edge servers.
The cloud server runs AWS Neptune \cite{ISWC2018:Neptune} to handle RDF data management, while the edge servers employ the open-source RDF store gStore \cite{VLDB2011:GStore,VLDBJ2014:gStore} to manage RDF graphs.
To evaluate scalability, the edge servers are configured with different hardware capacities. By default, four virtual machines are deployed as edge servers, each equipped with 1 CPU and 2 GB of memory. The environment also emulates 20 end users submitting SPARQL queries. In this setup, edge servers do not communicate directly with each other.

All codes are implemented in C++. The Gurobi Optimizer 10.0.1 is used to solve the linear programming problem via its C++ interface, which allows seamless integration with the main system. Additionally, network traffic between end users and servers is measured using iperf\footnote{https://iperf.fr/} and iftop\footnote{https://linux.die.net/man/8/iftop/}.

\textbf{Baseline.} Since no previous study addressed the problem, four baseline methods are designed for comparison. The first, Cloud-Only, processes all queries in the cloud. The second, Random, assigns queries randomly to either the cloud or capable edge servers. The third, Edge-First, prioritizes edge servers for query processing whenever they are capable, without considering the issue of computational resources allocation. The last, Greedy, allocates each query to the cloud or a capable edge server that minimizes execution cost.

{The problem studied in this work is a MINLP problem, and its objective function contains reciprocal nonlinear terms of the form $D_{n,k}\times e_{n,k} \times \frac{c_n}{f_{n,k}}$, which most mainstream MINLP solvers cannot handle directly. Even solvers that support such nonlinearities, such as Couenne \cite{belotti2009branching}, struggle to find feasible solutions within a reasonable time at practical scales. For instance, at the scale $|\mathcal{K}|=4, |\mathcal{N}|=20$, the solving time of Couenne exceeds 10 minutes. For these reasons, this study does not directly compare with general MINLP solvers. }

\begin{table}
\centering
\caption{Statistics of Datasets}
    \begin{tabular}{rrrr}
   \hline
    Dataset & $\#$Triples & $\#$Entities & $\#$Properties \\
     \hline
 WatDiv 100M  &  108,997,714   &  5,212,745 &  86  \\
 WatDiv 200M & 219,649,292  &  10,424,745  &  86  \\
 WatDiv 300M & 329,641,496  &  15,636,745  &  86  \\
 WatDiv 400M & 436,516,979  &  20,848,745  &  86  \\
 WatDiv 500M & 545,894,100  &  26,060,745  &  86  \\
  \hline
DBpedia  &  1,111,481,066   &  139,493,254  &  124,034  \\  
   \hline  \end{tabular}%
  \label{table:Datasets}
\end{table}

\textbf{Dataset.} To evaluate the performance of our methods,
we use a synthetic dataset, WatDiv \cite{DBLP:conf/semweb/AlucHOD14} and a real-world dataset, DBpedia \cite{Dataset:DBpedia}, whose statistics are shown in Table \ref{table:Datasets}.

WatDiv \cite{DBLP:conf/semweb/AlucHOD14} is a semantic data benchmark created by the Data Systems Group at the University of Waterloo. It consists of a data generator and a query generator that produce RDF datasets with diverse structures and a variety of query types.

DBpedia \cite{Dataset:DBpedia} extracts RDF data from Wikipedia’s extensive information. Multiple versions of DBpedia are available and regularly updated to reflect changes in Wikipedia content. In this paper, we extract data from version 3.4.

\textbf{Workload.} 
For DBpedia, we use queries from QALD \cite{Dataset:QALD} as the test workload. QALD is a benchmark designed for evaluating knowledge base question answering (KBQA) systems on DBpedia. It provides ground-truth SPARQL queries for benchmark natural language questions. We select a diverse set of these queries, covering various types and complexities. The selected queries are further used to generate templates that serve as patterns for constructing pattern-induced subgraphs stored at edge servers. For WatDiv, we use its query generator to generate queries as the workload.

Table \ref{table:queryResultSizeStatistics} shows the statistics of SPARQL query results sizes (in bytes) for the selected queries. In the edge-cloud environment, the distribution of end users varies: around 20\% of end users are connected to a single edge server, while the remaining users are connected to multiple edge servers. This distribution enables the evaluation of various network topologies and load conditions in real-world scenarios.

The scalability of the proposed method is evaluated using two metrics: query response time and assignment ratio. The query response time measures the total execution time of all user queries, including those processed at edge servers and in the cloud. The assignment ratio indicates the proportion of queries executed by edge servers or the cloud. We examine these metrics through two separate experiments to evaluate the scalability of both edge servers and end users.

\begin{table}[]
\centering
\caption{{Statistics of SPARQL Queries' Result Sizes }}
\begin{tabular}{cccc}
\hline
\multicolumn{1}{l}{Percentage} & \multicolumn{1}{l}{WatDiv (Bytes)} & \multicolumn{1}{l}{Percentage} & \multicolumn{1}{l}{Dbpedia (Bytes)} \\
\hline
23.33\% & <$10^{5}$ & 26.67\% & <$10^{5}$ \\
66.67\% & $10^{5}-10^{6}$ & 40.00\% & $10^{5}-10^{6}$  \\
6.67\% & $10^{6}-10^{7}$ & 23.33\% & $10^{6}-10^{8}$ \\
3.33\% & >$10^{7}$ & 10.00\% & >$10^{8}$ \\
\hline                  
\end{tabular}%
\label{table:queryResultSizeStatistics}
\end{table}

\subsection{Scalability Test of Edge Servers}
 As we set some virtual machines in servers as edge servers, their configurations can vary for scalability tests. Here, we test the scalability of our method in five aspects: varying storage capacity of edge servers, varying computing power of edge servers, varying bandwidths of edge servers, varying numbers of edge servers and varying graph sizes. By default, we refer to commonly used devices for edge servers, such as Raspberry Pi, and configure the edge servers with 2 GB storage and a 0.2 GHz CPU. The bandwidth from end users to edge servers is approximately 70-80 Mbps, while that from the end users to the cloud is approximately 5 Mbps.

\subsubsection{Varying Storage Capacity of Edge Servers}
In this experiment, we study the impact of the storage sizes of edge servers, ranging from 1.0GB to 2.5GB. Variations in the storage capacities affect $e_{n,k}$ values for each query $Q_n$ sent to edge server $k$. Larger storage allows edge servers to store subgraphs covering more workload patterns, increasing the assignment ratio and improving query execution efficiency. 
Fig. \ref{fig:ExpStorageSizesTime} shows query response times, and Table \ref{table:CPUSpeedsOffloadingRate} presents the corresponding assignment ratio. Except for the Cloud-Only solution, response time generally decreases as storage capacity increases. The Random strategy exhibits some fluctuation due to its lack of data placement awareness. When storage is limited (e.g., 1.0GB), some queries must be assigned to the cloud for execution, the assignment ratio of the cloud is Relatively high.
As storage capacity increases, more queries can be processed at the edge, leading to higher assignment ratios and lower response times. Particularly after reaching a medium storage configuration, our method starts to outperform other strategies, demonstrating more efficient query scheduling, with a performance improvement of 15.46\% in response time.

\begin{table}[]
\centering
\caption{{Assignment Ratio for Different Storage Capacity}}
\begin{tabular}{rrrrrrrrr}
\hline
 & \multicolumn{4}{c}{WatDiv} & \multicolumn{4}{c}{DBpedia}\\
\cline{2-9}
 & 1.0GB  & 1.5GB  & 2.0GB  & 2.5GB   & 1.0GB  & 1.5GB  & 2.0GB  & 2.5GB  \\
\hline
$ES_1$ &   8.50\% & 21.00\% & 18.50\% & 22.50\% & 23.00\% & 27.50\% & 28.00\% & 27.00\% \\
$ES_2$ &  13.00\% & 25.00\% & 24.50\% & 24.00\% & 17.50\% & 23.50\% & 23.00\% & 25.00\% \\
$ES_3$ &  22.50\% & 24.50\% & 27.50\% & 24.50\% & 18.50\% & 23.00\% & 28.00\% & 20.50\% \\
$ES_4$ &  17.00\% & 24.50\% & 26.00\% & 26.50\% & 18.00\% & 20.00\% & 12.50\% & 18.50\% \\
$Cloud$&  39.00\% &  5.00\% & 3.50\% &   2.50\% &  23.00\% & 6.00\% &  8.50\% &  9.00\% \\
\hline
\end{tabular}%
\label{table:StorageSizesOffloadingRate}
\end{table}

\begin{figure}[]
\centering
    \subfigure[{WatDiv}]{%
        \centering
		\resizebox{0.47\columnwidth}{!}{
	\definecolor{lightBlue}{HTML}{42B4E7}
\definecolor{lightGrey}{HTML}{336699}
\definecolor{lightGreen}{HTML}{73B66B}
\definecolor{yellow}{HTML}{FFCB18}
\begin{tikzpicture}[]
    \begin{axis}[
        ybar,
        width=8cm,
        height=6cm,
        bar width=6pt,
        xlabel=Storage Capacity,
        ylabel=Query Response Time (in ms),
        xtick=data,
        xticklabels={1.0GB,1.5GB,2.0GB,2.5GB},
        ymax=5000000,
        ymode=log,
        enlarge x limits=0.2,
        legend cell align=left,
        legend style={draw=none},
        legend pos=north west,
        legend columns=2
    ]

    \addplot+[ybar, fill=lightBlue, draw=black, postaction={
        pattern=north east lines,
        pattern color=black
    }] coordinates {
        (5,     560313.9)
        (10,    560313.9)
        (15,    560313.9)
        (20,    560313.9)
    };
    \addlegendentry{Cloud-Only};

    \addplot+[ybar, fill=lightGrey, draw=black] coordinates {
        (5,     288465.2)
        (10,    179525.8)
        (15,    85137)
        (20,    132180.3)
    };
    \addlegendentry{Random};

    \addplot+[ybar, fill=lightGreen, draw=black, postaction={
        pattern=north west lines,
        pattern color=black 
    }] coordinates {
        (5,     127858.7)
        (10,    100077.2)
        (15,    24325)
        (20,    21889.5)
    };
    \addlegendentry{Edge-First};

    \addplot+[ybar, fill=orange, draw=black] coordinates {
        (5,     122440.9)
        (10,    94424.4)
        (15,    22129.9)
        (20,    20293.6)
    };
    \addlegendentry{Greedy};

    \addplot+[ybar, fill=yellow, draw=black] coordinates {
        (5,     119632.9)
        (10,    90971.8 )
        (15,    20563.4)
        (20,    18579.4)
    };
    \addlegendentry{Ours};
 
    \end{axis}
\end{tikzpicture}
    }
       \label{fig:ExpStorageSizesTimeWatDiv}%
       }
   \subfigure[{DBpedia}]{%
        \centering
		\resizebox{0.47\columnwidth}{!}{
	\definecolor{lightBlue}{HTML}{42B4E7}
\definecolor{lightGrey}{HTML}{336699}
\definecolor{lightGreen}{HTML}{73B66B}
\definecolor{yellow}{HTML}{FFCB18}
\begin{tikzpicture}[]
    \begin{axis}[
        ybar,
        width=8cm,
        height=6cm,
        bar width=6pt,
        xlabel=Storage Capacity,
        ylabel=Query Response Time (in ms),
        xtick=data,
        xticklabels={1.0GB,1.5GB,2.0GB,2.5GB},
        ymax=5000000,
        ymode=log,
        enlarge x limits=0.2,
        legend cell align=left,
        legend style={draw=none},
        legend pos=north west,
        legend columns=2 
    ]
    
    \addplot+[ybar, fill=lightBlue, draw=black, postaction={
        pattern=north east lines,
        pattern color=black
    }] coordinates {
        (5,     1078871.7)
        (10,    1078871.7)
        (15,    1078871.7)
        (20,    1078871.7)
    };
    \addlegendentry{Cloud-Only};
    
    \addplot+[ybar, fill=lightGrey, draw=black] coordinates {
        (5,     776665.4)
        (10,    621031.6)
        (15,    506738.3)
        (20,    492541)
    };
    \addlegendentry{Random};
    
    \addplot+[ybar, fill=lightGreen, draw=black, postaction={
        pattern=north west lines,
        pattern color=black 
    }] coordinates {
        (5,     669413.5)
        (10,    234014.1)
        (15,    192703.3)
        (20,    185990.9)
    };
    \addlegendentry{Edge-First};
    
    \addplot+[ybar, fill=orange, draw=black] coordinates {
        (5,     669413.5)
        (10,    205262.2)
        (15,    182857.8)
        (20,    177833.9)
    };
    \addlegendentry{Greedy};
    
      \addplot+[ybar, fill=yellow, draw=black] coordinates {
        (5,     669413.5)
        (10,    202306.7 )
        (15,    176529.2)
        (20,    173479.5)
    };
    \addlegendentry{Ours};
 
    \end{axis}
\end{tikzpicture}
    }
       \label{fig:ExpStorageSizesTimeDBpedia}%
       }
\caption{Varying Storage Sizes of Edge Servers}%
 \label{fig:ExpStorageSizesTime}
\end{figure}
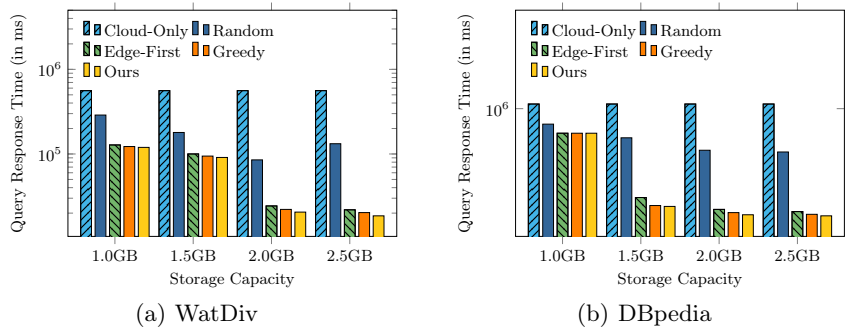

\subsubsection{Varying Computing Power of Edge Servers}

This experiment evaluate the impact of edge server computing power. We vary the CPU clock speed of each edge server from 0.2GHz to 0.8GHz, in 0.2GHz increments. Different CPU clock speeds of edge servers reflect varying computational resources. Therefore, differences in computing power will influence the $F_k$ values for a given edge server $k$. 
Fig. \ref{fig:ExpCPUSpeeds} and Table \ref{table:CPUSpeedsOffloadingRate} present the query response times and assignment ratios under different computing power settings. As computing power increases, the concurrent processing capability of edge servers improves, thereby reducing the overall query response time. Our method exhibits lower latency than the others. Higher computing power allows for handling more complex queries, thus increasing the assignment ratio to edge servers. 
Variations in assignment ratios across servers are mainly due to differences in query complexity.

\begin{table}[]
\centering
\caption{{Assignment Ratio for Different Computing Power}}
\begin{tabular}{rrrrrrrrr}
\hline
 & \multicolumn{4}{c}{WatDiv} & \multicolumn{4}{c}{DBpedia}\\
\cline{2-9}
 & 0.2GHz  & 0.4GHz  & 0.6GHz  & 0.8GHz & 0.2GHz  & 0.4GHz  & 0.6GHz  & 0.8GHz \\
\hline
$ES_1$  & 18.50\% & 19.50\% & 19.00\% & 19.00\% & 28.00\% & 28.50\% & 29.00\% & 31.00\% \\
$ES_2$  & 24.50\% & 25.00\% & 25.00\% & 23.50\% & 23.00\% & 21.00\% & 21.00\% & 22.50\% \\
$ES_3$  & 27.50\% & 27.50\% & 27.50\% & 30.00\% & 28.00\% & 28.00\% & 28.50\% & 29.00\% \\
$ES_4$  & 26.00\% & 26.00\% & 27.50\% & 27.50\% & 12.50\% & 15.50\% & 16.00\% & 16.50\% \\
$Cloud$ &  3.50\% &  2.00\% &  1.00\% &  0.00\% &  8.50\% &  7.00\% &  5.50\% &  1.00\%\\
\hline
\end{tabular}%
\label{table:CPUSpeedsOffloadingRate}
\end{table}

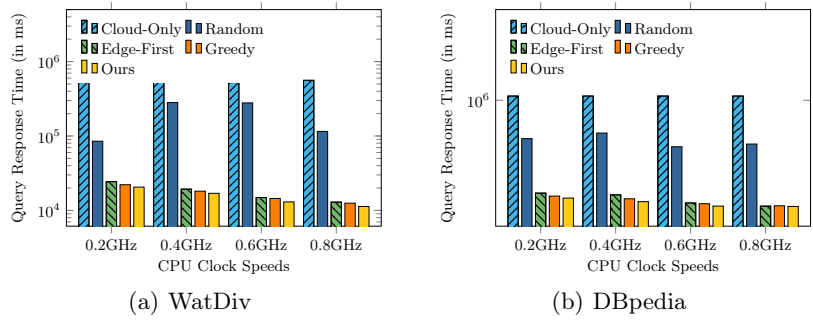
\begin{figure}
\centering
\subfigure[{WatDiv}]{%
        \centering
		\resizebox{0.45\columnwidth}{!}{
	\definecolor{lightBlue}{HTML}{42B4E7}
\definecolor{lightGrey}{HTML}{336699}
\definecolor{lightGreen}{HTML}{73B66B}
\definecolor{yellow}{HTML}{FFCB18}
\begin{tikzpicture}[]
    \begin{axis}[
        ybar,
        width=8cm,
        height=6cm,
        bar width=6pt,
        xlabel=CPU Clock Speeds,
        ylabel=Query Response Time (in ms),
        xtick=data,
        xticklabels={0.2GHz,0.4GHz,0.6GHz,0.8GHz},
        ymax=5000000,
        ymode=log,
        enlarge x limits=0.2,
        legend cell align=left,
        legend style={draw=none},
        legend pos=north west,
        legend columns=2 
    ]
    
    \addplot plot[ybar, fill=lightBlue, draw=black, postaction={
        pattern=north east lines,
        pattern color=black
    }] coordinates {
        (2,   560313.9)
        (4,   560313.9)
        (6,   560313.9)
        (8,   560313.9)
    };
    \addlegendentry{Cloud-Only};
    
    \addplot plot[ybar, fill=lightGrey, draw=black] coordinates {
        (2,     85137)
        (4,       281771.6)
        (6,     278560.4)
        (8,     115241.9)
    };
    \addlegendentry{Random};
    
    \addplot plot[ybar, fill=lightGreen, draw=black, postaction={
        pattern=north west lines,
        pattern color=black 
    }] coordinates {
        (2,     24325)
        (4,      19316.2)
        (6,      14891.6)
        (8,      12927.4)
    };
    \addlegendentry{Edge-First};
    
    \addplot plot[ybar, fill=orange, draw=black] coordinates {
        (2,     22129.9)
        (4,      18125)
        (6,      14475.8)
        (8,      12511.9)
    };
    \addlegendentry{Greedy};
    
    \addplot plot[ybar, fill=yellow, draw=black] coordinates {
        (2,     20563.4)
        (4,       16972.8)
        (6,     13019.3)
        (8,     11294.4)
    };
    \addlegendentry{Ours};
 
    \end{axis}
\end{tikzpicture}
    }
       \label{fig:ExpCPUSpeedsWatDiv}%
       }
   \subfigure[{DBpedia}]{%
        \centering
		\resizebox{0.45\columnwidth}{!}{
	\definecolor{lightBlue}{HTML}{42B4E7}
\definecolor{lightGrey}{HTML}{336699}
\definecolor{lightGreen}{HTML}{73B66B}
\definecolor{yellow}{HTML}{FFCB18}
\begin{tikzpicture}[]
    \begin{axis}[
        ybar,
        width=8cm,
        height=6cm,
        bar width=6pt,
        xlabel=CPU Clock Speeds,
        ylabel=Query Response Time (in ms),
        xtick=data,
        xticklabels={0.2GHz,0.4GHz,0.6GHz,0.8GHz},
        ymax=5000000,
        ymode=log,
        enlarge x limits=0.2,
        legend cell align=left,
        legend style={draw=none},
        legend pos=north west,
        legend columns=2 
    ]
    
    \addplot plot[ybar, fill=lightBlue, draw=black, postaction={
        pattern=north east lines,
        pattern color=black
    }] coordinates {
        (2,     1078871.7)
        (4,     1078871.7)
        (6,     1078871.7)
        (8,     1078871.7)
    };
    \addlegendentry{Cloud-Only};
    
    \addplot plot[ybar, fill=lightGrey, draw=black] coordinates {
        (2,    506738.3)
        (4,   559525.2)
        (6,    438645)
        (8,    460336)
    };
    \addlegendentry{Random};
    
    \addplot plot[ybar, fill=lightGreen, draw=black, postaction={
        pattern=north west lines,
        pattern color=black 
    }] coordinates {
        (2,    192703.3)
        (4,    186844.2 )
        (6,    161810.9)
        (8,    152921.6)
    };
    \addlegendentry{Edge-First};
    
    \addplot plot[ybar, fill=orange, draw=black] coordinates {
        (2,    182857.8)
        (4,    174286.8 )
        (6,    159636.2)
        (8,    153925.4)
    };
    \addlegendentry{Greedy};
    
    \addplot plot[ybar, fill=yellow, draw=black] coordinates {
        (2,     176529.2)
        (4,     165635.7)
        (6,    153069.3)
        (8,    152037.7)
    };
    \addlegendentry{Ours};
 
    \end{axis}
\end{tikzpicture}
    }
       \label{fig:ExpCPUSpeedsDBpedia}%
       }
\caption{Varying Computing Power of Edge Servers}%
 \label{fig:ExpCPUSpeeds}
\end{figure}

\begin{table}[] 
\centering
\caption{{Assignment Ratio for Different Bandwidths}}
\begin{tabular}{rrrrrrrrr}
\hline
 & \multicolumn{4}{c}{WatDiv} & \multicolumn{4}{c}{DBpedia}\\
\cline{2-9}
& 10Mbps& 30Mbps& 50Mbps& 70Mbps & 10Mbps& 30Mbps& 50Mbps& 70Mbps \\
\hline
$ES_1$ & 18.00\%& 18.00\%& 18.50\% & 18.50\% & 23.00\%& 26.00\%& 28.50\% & 28.00\%  \\
$ES_2$ & 25.00\%& 23.50\%& 24.50\% & 24.50\% & 20.50\%& 20.50\%& 21.00\% & 23.00\%  \\
$ES_3$ & 25.00\%& 26.00\%& 27.50\% & 27.50\% & 24.00\%& 25.50\%& 27.50\% & 28.00\%  \\
$ES_4$ & 24.00\%& 25.50\%& 26.00\% & 26.00\% & 12.00\%& 12.50\%& 12.50\% & 12.50\%  \\
$Cloud$&  8.00\%&  7.00\%&  3.50\% &  3.50\% & 20.50\%& 15.50\%& 10.50\% &  8.50\% \\
\hline
\end{tabular}%
\label{table:BandwidthsOffloadingRate}
\end{table}

\begin{figure}
\centering
\subfigure[{WatDiv}]{%
        \centering
		\resizebox{0.45\columnwidth}{!}{
	\definecolor{lightBlue}{HTML}{42B4E7}
\definecolor{lightGrey}{HTML}{336699}
\definecolor{lightGreen}{HTML}{73B66B}
\definecolor{yellow}{HTML}{FFCB18}
\begin{tikzpicture}[]
    \begin{axis}[
        ybar,
        width=8cm,
        height=6cm,
        bar width=6pt,
        xlabel=Bandwidths,
        ylabel=Query Response Time (in ms),
        xtick=data,
        xticklabels={10Mbps,30Mbps,50Mbps,70Mbps},
        ymax=5000000,
        ymode=log,
        enlarge x limits=0.2,
        legend cell align=left,
        legend style={draw=none},
        legend pos=north west,
        legend columns=2 
    ]
    
    \addplot plot[ybar, fill=lightBlue, draw=black, postaction={
        pattern=north east lines,
        pattern color=black
    }] coordinates {
        (2,     560313.9)
        (4,    560313.9)
        (6,    560313.9)
        (8,    560313.9)
    };
    \addlegendentry{Cloud-Only};
    
    \addplot plot[ybar, fill=lightGrey, draw=black] coordinates {
        (2,     144956.9)
        (4,    104347.3)
        (6,    90734.5)
        (8,    86880)
    };
    \addlegendentry{Random};
    
    \addplot plot[ybar, fill=lightGreen, draw=black, postaction={
        pattern=north west lines,
        pattern color=black 
    }] coordinates {
        (2,     40420.2)
        (4,    31011.6)
        (6,    28223.3)
        (8,    23280.3)
    };
    \addlegendentry{Edge-First};
    
    \addplot plot[ybar, fill=orange, draw=black] coordinates {
        (2,     37852.7)
        (4,      29153.3)
        (6,      27316.1)
        (8,      21899.5)
    };
    \addlegendentry{Greedy};
    
    \addplot plot[ybar, fill=yellow, draw=black] coordinates {
        (2,     36975.9)
        (4,   28260.4)
        (6,    26048.4)
        (8,    19051.1)
    };
    \addlegendentry{Ours};
 
    \end{axis}
\end{tikzpicture}
    }
       \label{fig:ExpBandwidthsWatDiv}%
       }
   \subfigure[{DBpedia}]{%
        \centering
		\resizebox{0.45\columnwidth}{!}{
	\definecolor{lightBlue}{HTML}{42B4E7}
\definecolor{lightGrey}{HTML}{336699}
\definecolor{lightGreen}{HTML}{73B66B}
\definecolor{yellow}{HTML}{FFCB18}
\begin{tikzpicture}[]
    \begin{axis}[
        ybar,
        width=8cm,
        height=6cm,
        bar width=6pt,
        xlabel=Bandwidths,
        ylabel=Query Response Time (in ms),
        xtick=data,
        xticklabels={10Mbps,30Mbps,50Mbps,70Mbps},
        ymax=5000000,
        ymode=log,
        enlarge x limits=0.2,
        legend cell align=left,
        legend style={draw=none},
        legend pos=north west,
        legend columns=2 
    ]
    
    \addplot plot[ ybar, fill=lightBlue, draw=black, postaction={
        pattern=north east lines,
        pattern color=black
    }] coordinates {
        (2,     1078871.7)
        (4,   1078871.7 )
        (6,    1078871.7)
        (8,    1078871.7)
    };
    \addlegendentry{Cloud-Only};
    
    \addplot plot[ybar, fill=lightGrey, draw=black] coordinates {
        (2,     545575.3)
        (4,    351079)
        (6,    325827.6)
        (8,    506738.3)
    };
    \addlegendentry{Random};
    
    \addplot plot[ ybar, fill=lightGreen, draw=black, postaction={
        pattern=north west lines,
        pattern color=black 
    }] coordinates {
        (2,     457582.8)
        (4,    232430.7)
        (6,    195743.6)
        (8,    192703.3)
    };
    \addlegendentry{Edge-First};
    
    \addplot plot[ybar, fill=orange, draw=black] coordinates {
        (2,     443156.3)
        (4,      229306.5)
        (6,      189372.4)
        (8,      182857.8)
    };
    \addlegendentry{Greedy};
    
    \addplot plot[ ybar, fill=yellow, draw=black] coordinates {
        (2,     412761.5)
        (4,    220567.5)
        (6,    184227.5)
        (8,    176529.2)
    };
    \addlegendentry{Ours};
 
    \end{axis}
\end{tikzpicture}
    }
       \label{fig:ExpBandwidthsDBpedia}%
       }
\caption{Varying Bandwidths of Edge Servers}%
 \label{fig:ExpBandwidths}
\end{figure}
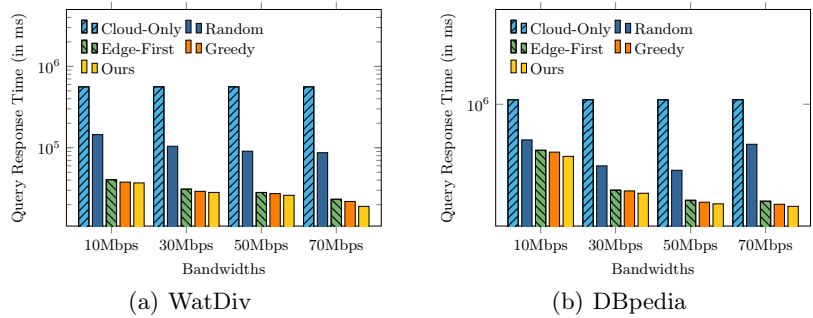

\subsubsection{Varying Bandwidths of Edge Servers}
This experiment evaluates the impact of network bandwidth between end users and edge servers. Using the Linux traffic control tool $tc$, we increase the network bandwidth from 10Mbps to 70Mbps, with an increments of 20Mbps. 
Changes in bandwidth affect the $r_{n,k}$ value between end user $n$ and edge server $k$, thereby influencing data transmission delay and overall response time. 
Fig. \ref{fig:ExpBandwidths} and Table \ref{table:BandwidthsOffloadingRate} show the query response times and the assignment ratios under different bandwidth settings. Lower network bandwidth (e.g., 10Mbps) increases data transmission time, leading to higher overall processing latency. To avoid this, more queries are assigned to the cloud, increasing the assignment ratio in the cloud. As the bandwidth increases, the query response time gradually decreases, and the assignment ratio to edge servers gradually increases. Compared to other methods, our method maintains better performance under fluctuating bandwidth conditions by dynamically adjusting the query assignment strategy, achieving up to a 13.01\%-18.17\% reduction in response time under higher bandwidth conditions.

{\subsubsection{Varying Number of Edge Servers and End Users} 
We scale the number of edge servers and end users proportionally to evaluate the system behavior at larger scales. Specifically, the number of edge servers is scaled from 4 to 128, while the number of end users is increased proportionally from 20 to 640. The experiments are conducted on both the WatDiv and DBpedia datasets, and the results are shown in the Fig. 10. 
As the system scale increases, the query response times of all methods exhibit an upward trend.
Although the increase in the number of end users with system scaling naturally reduces the assignment ratio per server, our method is still able to maintain a relatively stable load distribution across all scale settings.
Across both datasets, our method consistently achieves better performance under all scale configurations. As the scale grows, its advantage over the edge-first and greedy baselines becomes more evident, yielding performance improvements of of 16.20\% and 6.47\%, respectively, under large-scale settings on the WatDiv dataset.
In scale reaches beyond (32, 160), the cloud-only approach experiences timeouts, indicating that edge-based solutions are more scalable.}

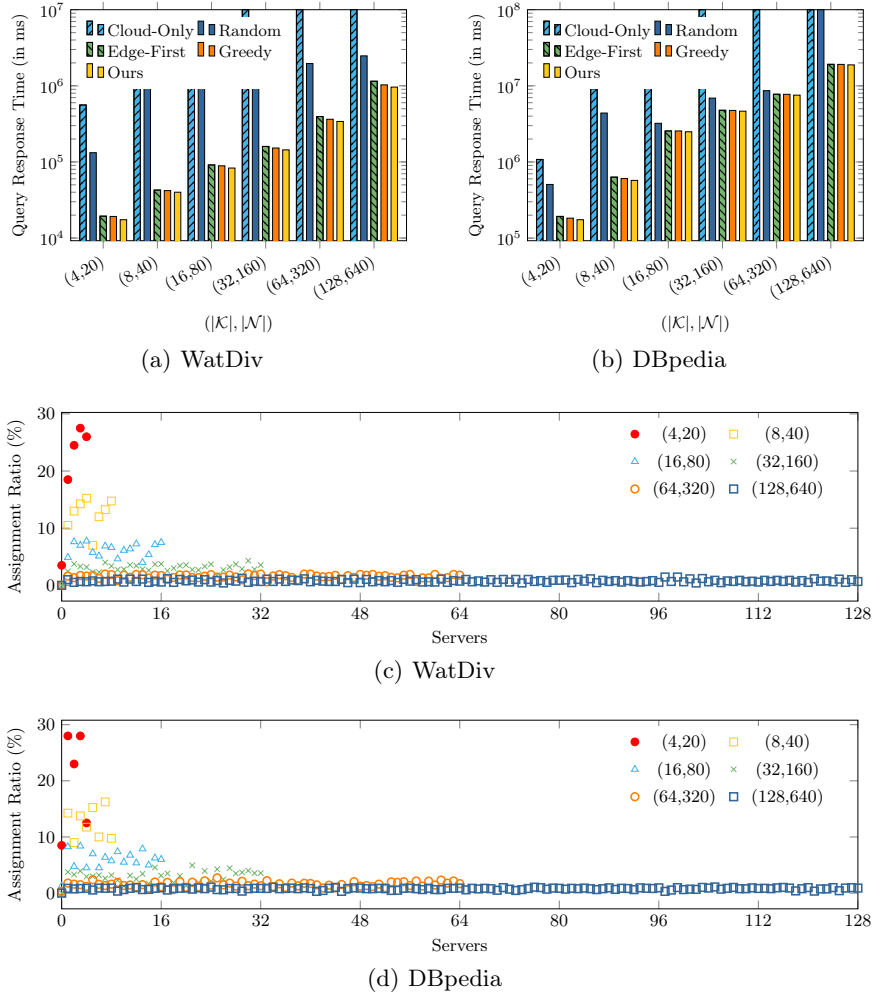
\begin{figure}
\centering
\subfigure[{WatDiv}]{%
        \centering
		\resizebox{0.48\columnwidth}{!}{
	\definecolor{lightBlue}{HTML}{42B4E7}
\definecolor{lightGrey}{HTML}{336699}
\definecolor{lightGreen}{HTML}{73B66B}
\definecolor{yellow}{HTML}{FFCB18}
\begin{tikzpicture}[]
    \begin{axis}[
        ybar,
        width=8cm,
        height=6cm,
        bar width=3.5pt,
        xlabel={($|\mathcal{K}|, |\mathcal{N}|$)},,
        ylabel=Query Response Time (in ms),
        xtick=data,
        xticklabel style={rotate=30,anchor=east,yshift=-2mm},
        xticklabels={(4,20), (8,40), (16,80), (32,160), (64,320), (128,640)},
        ymax=10000000,
        ymode=log,
        enlarge x limits=0.12,
        legend cell align=left,
        legend style={draw=none},
        legend pos=north west,
        legend columns=2 
    ]
    
    \addplot plot[ ybar, fill=lightBlue, draw=black, postaction={
        pattern=north east lines,
        pattern color=black
    }] coordinates {
        (0,     560313.9)
        (1,    1535486)
        (2,    6668673)
        (3,    10000000)
        (4,    10000000)
        (5,    10000000)
    };
    \addlegendentry{Cloud-Only};
    
    \addplot plot[ybar, fill=lightGrey, draw=black] coordinates {
        (0,     132180.3)
        (1,    929263.4)
        (2,    1688058.25)
        (3,    2559606.25)
        (4,    1962617.5) 
        (5,    2475260)
    };
    \addlegendentry{Random};
    
    \addplot plot[ ybar, fill=lightGreen, draw=black, postaction={
        pattern=north west lines,
        pattern color=black 
    }] coordinates {
        (0,     19396.6)
        (1,    42708)
        (2,    91356.6)
        (3,    159813.3)
        (4,    394800.3)
        (5,    1148801.5)
    };
    \addlegendentry{Edge-First};
    
    \addplot plot[ybar, fill=orange, draw=black] coordinates {
        (0,     19229.8)
        (1,    41996.5)
        (2,    88928.3)
        (3,    152424.5)
        (4,    363672)
        (5,    1029226)
    };
    \addlegendentry{Greedy};
    
    \addplot plot[ ybar, fill=yellow, draw=black] coordinates {
        (0,     17415.8)
        (1,    40010.7)
        (2,    83123.9)
        (3,    144034.5)
        (4,    340419.1)
        (5,    962639.5)
    };
    \addlegendentry{Ours};
 
    \end{axis}
 
\end{tikzpicture}
    }
       }
\subfigure[{DBpedia}]{%
        \centering
		\resizebox{0.48\columnwidth}{!}{
	\definecolor{lightBlue}{HTML}{42B4E7}
\definecolor{lightGrey}{HTML}{336699}
\definecolor{lightGreen}{HTML}{73B66B}
\definecolor{yellow}{HTML}{FFCB18}
\begin{tikzpicture}[]
    \begin{axis}[
        ybar,
        width=8cm,
        height=6cm,
        bar width=3.5pt,
        xlabel={($|\mathcal{K}|, |\mathcal{N}|$)},
        ylabel=Query Response Time (in ms),
        xtick=data,
        xticklabel style={rotate=30,anchor=east,yshift=-2mm},
        xticklabels={(4,20), (8,40), (16,80), (32,160), (64,320), (128,640)},
        ymax=100000000,
        ymode=log,
        enlarge x limits=0.12,
        legend cell align=left,
        legend style={draw=none},
        legend pos=north west,
        legend columns=2 
    ]
    
    \addplot plot[ ybar, fill=lightBlue, draw=black, postaction={
        pattern=north east lines,
        pattern color=black
    }] coordinates {
        (0,     1078871.7)
        (1,    9165411)
        (2,    27488635)
        (3,    100000000)
        (4,    100000000)
        (5,    100000000)
    };
    \addlegendentry{Cloud-Only};
    
    \addplot plot[ybar, fill=lightGrey, draw=black] coordinates {
        (0,     506738.3)
        (1,    4391761)
        (2,    3221668)
        (3,    6901705)
        (4,    8636554) 
        (5,    100000000)
    };
    \addlegendentry{Random};
    
    \addplot plot[ ybar, fill=lightGreen, draw=black, postaction={
        pattern=north west lines,
        pattern color=black 
    }] coordinates {
        (0,     192703.3)
        (1,    632060.2)
        (2,    2554702.8)
        (3,    4774382.7)
        (4,    7755848.6)
        (5,    19166636)
    };
    \addlegendentry{Edge-First};
    
    \addplot plot[ybar, fill=orange, draw=black] coordinates {
        (0,     182857.8)
        (1,    605949.7)
        (2,    2555527.5)
        (3,    4756479.5)
        (4,    7739794.1)
        (5,    19134049)
    };
    \addlegendentry{Greedy};
    
    \addplot plot[ ybar, fill=yellow, draw=black] coordinates {
        (0,     174393.4)
        (1,    572657.8)
        (2,    2493304.8)
        (3,    4646170)
        (4,    7549755.9)
        (5,    18855915)
    };
    \addlegendentry{Ours};
 
    \end{axis}
 
\end{tikzpicture}
    }
       }
\subfigure[{WatDiv}]{%
        \centering
		\resizebox{\columnwidth}{!}{
	\definecolor{lightBlue}{HTML}{42B4E7}
\definecolor{lightGrey}{HTML}{336699}
\definecolor{lightGreen}{HTML}{73B66B}
\definecolor{yellow}{HTML}{FFCB18}
\begin{tikzpicture}[]
\begin{axis}[
    width=16cm,       
    height=5cm,
    xlabel={Servers},
    ylabel={Assignment Ratio (\%)},
    xmin=0, xmax=128,
    xtick={0,16,32,48,64,80,96,112,128},
    legend style={draw=none, column sep=5pt, row sep=2pt},
    legend columns=2,
    legend pos=north east,
    unbounded coords=discard
]

\addplot[red, mark=*, only marks] table [x=server, y=y4, col sep=comma] {dataWatdiv.csv};
\addlegendentry{(4,20)}

\addplot[yellow, mark=square, only marks] table [x=server, y=y8, col sep=comma] {dataWatdiv.csv};
\addlegendentry{(8,40)}

\addplot[lightBlue, mark=triangle, only marks] table [x=server, y=y16, col sep=comma] {dataWatdiv.csv};
\addlegendentry{(16,80)}

\addplot[lightGreen, mark=x, only marks] table [x=server, y=y32, col sep=comma] {dataWatdiv.csv};
\addlegendentry{(32,160)}

\addplot[orange, mark=o, thick, only marks] table [x=server, y=y64, col sep=comma] {dataWatdiv.csv};
\addlegendentry{(64,320)}

\addplot[lightGrey, mark=square, thick, only marks] table [x=server, y=y128, col sep=comma] {dataWatdiv.csv};
\addlegendentry{(128,640)}

\end{axis}
\end{tikzpicture}
    }
       }
\subfigure[{DBpedia}]{%
        \centering
		\resizebox{\columnwidth}{!}{
	\definecolor{lightBlue}{HTML}{42B4E7}
\definecolor{lightGrey}{HTML}{336699}
\definecolor{lightGreen}{HTML}{73B66B}
\definecolor{yellow}{HTML}{FFCB18}
\begin{tikzpicture}[]
\begin{axis}[
    width=16cm,       
    height=5cm,
    xlabel={Servers},
    ylabel={Assignment Ratio (\%)},
    xmin=0, xmax=128,
    xtick={0,16,32,48,64,80,96,112,128},
    legend style={draw=none, column sep=5pt, row sep=2pt},
    legend columns=2,
    legend pos=north east,
    unbounded coords=discard,
]

\addplot[red, mark=*, only marks] table [x=server, y=y4, col sep=comma] {dataDbpedia.csv};
\addlegendentry{(4,20)}

\addplot[yellow, mark=square, only marks] table [x=server, y=y8, col sep=comma] {dataDbpedia.csv};
\addlegendentry{(8,40)}

\addplot[lightBlue, mark=triangle, only marks] table [x=server, y=y16, col sep=comma] {dataDbpedia.csv};
\addlegendentry{(16,80)}

\addplot[lightGreen, mark=x, only marks] table [x=server, y=y32, col sep=comma] {dataDbpedia.csv};
\addlegendentry{(32,160)}

\addplot[orange, mark=o, thick, only marks] table [x=server, y=y64, col sep=comma] {dataDbpedia.csv};
\addlegendentry{(64,320)}

\addplot[lightGrey, mark=square, thick, only marks] table [x=server, y=y128, col sep=comma] {dataDbpedia.csv};
\addlegendentry{(128,640)}

\end{axis}
\end{tikzpicture}
    }
       } 
{\caption{Scalability evaluation}}%
 \label{fig:ExpEdgeServerNumber}
\end{figure}

\subsubsection{Varying Graph Sizes} 
{
We generate three WatDiv datasets ranging from 100 million to 500 million triples. Fig. 11 and Table 8 show the query response times and assignment ratios for different graph sizes. 
As the graph size increases, the complexity and resource demands of query processing grow accordingly, leading to longer query response times. When scaling from 100MB to 200MB, the increased load on edge servers causes some lightweight queries to be assigned to the cloud, leading to a higher cloud assignment ratio. However, as the graph grows, the fraction of simple queries decreases, and the remaining complex queries incur higher communication and computation overhead in the cloud. Consequently, the system tends to retain them at edge servers despite increased local load, thus reducing the cloud assignment ratio.}

\begin{table}[]
\centering
{\caption{{Assignment Ratio for Different Graph Sizes}}}
\begin{tabular}{rrrrrr}
\hline
& 100M  & 200M   & 300M   & 400M & 500M \\
\hline
$ES_1$  & 18.50\% & 15.00\% & 17.50\% & 19.00\% & 19.00\%\\
$ES_2$  & 24.50\% & 23.00\% & 24.50\% & 23.00\% & 21.00\%\\
$ES_3$  & 27.50\% & 26.00\% & 27.00\% & 27.50\% & 25.00\%\\
$ES_4$  & 26.00\% & 24.50\% & 25.00\% & 30.00\% & 35.00\%\\
$Cloud$ &  3.50\% & 11.50\% &  6.00\% & 0.50\% & 0.00\%\\
\hline
\end{tabular}%
\label{table:NumberOffloadingRateGraphSizes}
\end{table}

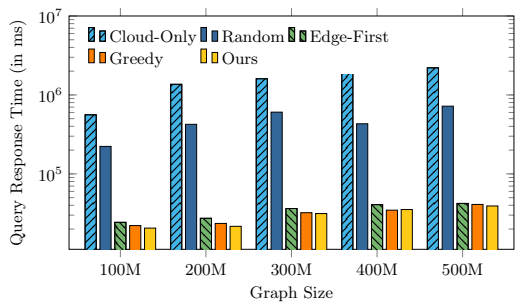
\begin{figure}[]
\centering
		\resizebox{0.6\columnwidth}{!}{
	\definecolor{lightBlue}{HTML}{42B4E7}
\definecolor{lightGrey}{HTML}{336699}
\definecolor{lightGreen}{HTML}{73B66B}
\definecolor{yellow}{HTML}{FFCB18}
\begin{tikzpicture}[]
    \begin{axis}[
        width=10cm,
        xtick=data,
        height = 6cm,
        ybar,
        bar width=6pt,
        xlabel=Graph Size,
        ylabel=Query Response Time (in ms),
        xtick=data,
        xticklabels={100M, 200M, 300M, 400M, 500M},
        ymax=10000000,
        ymode=log,
        enlarge x limits=0.15,
        legend cell align=left,
        legend style={draw=none},
        legend pos=north west,
        legend columns=3
    ]
    
    \addplot plot[ ybar, fill=lightBlue, draw=black, postaction={
        pattern=north east lines,
        pattern color=black
    }] coordinates {
        (5,     560313.9)
        (10,    1363553.2)
        (15,    1603032.8)
        (20,    1883050.3)
        (25,    2205401)
    };
    \addlegendentry{Cloud-Only};
    
    \addplot plot[ybar, fill=lightGrey, draw=black] coordinates {
        (5,     222856)
        (10,    424418.7)
        (15,    605276.7)
        (20,    431135)
        (25,    720234)
    };
    \addlegendentry{Random};
    
    \addplot plot[ybar, fill=lightGreen, draw=black, postaction={
        pattern=north west lines,
        pattern color=black 
    }] coordinates {
        (5,     24325)
        (10,    27364)
        (15,    36415.9)
        (20,    40746.7)
        (25,    42229.95)
    };
    \addlegendentry{Edge-First};
    
    \addplot plot[ybar, fill=orange, draw=black] coordinates {
        (5,     22129.9)
        (10,      23563.6)
        (15,      32226.2)
        (20,      34694.7)
        (25,    40998.15)
    };
    \addlegendentry{Greedy};
    
    \addplot plot[ybar, fill=yellow, draw=black] coordinates {
        (5,     20563.4)
        (10,    21655.7)
        (15,    31422.6)
        (20,    35326.5)
        (25,    39213.6)
    };
    \addlegendentry{Ours};
 
    \end{axis}
\end{tikzpicture}
    }
{\caption{Varying Graph Sizes}}%
 \label{fig:ExpGraphSizes}
\end{figure}

\subsection{Scalability Test of End Users}
In this section, we focus on the impact of the number of end users and the complexity of queries on system performance. The experiments explore how increased user demand impacts assignment ratio and response time.

\subsubsection{Varying Number of Queries Per End User}
{In the default setting, the total number of queries equals the number of end users, where each end user issues one query request. We also study the case where each end user issues multiple queries while keeping the number of end users fixed. We consider scenarios where each user submits 1 to 4 queries. 
Fig. 12 and Table 9 report the query response time and assignment ratios, respectively.
As the number of queries per end user increases, the overall workload grows, resulting in a consistent increase in query response time.
For both datasets, most queries are assigned to edge servers, exhibiting a high edge assignment ratio. Although the assignment ratio to the cloud remains relatively small, the absolute number of queries assigned to the cloud increases as the total workload grows.}

\begin{table}[]
\centering
{\caption{{Assignment Ratio for Varying Number of Queries}}}
\begin{tabular}{rrrrrrrrr}
\hline
 & \multicolumn{4}{c}{WatDiv} & \multicolumn{4}{c}{DBpedia}\\
\cline{2-9}
& 1  & 2  & 3  & 4 & 1  & 2  & 3  & 4  \\
\hline
$ES_1$ & 18.50\% &17.50\% & 18.20\% & 17.50\% & 28.00\% & 28.25\% & 29.67\% & 28.12\% \\
$ES_2$ & 24.50\% & 27.00\% & 23.50\% & 23.62\%& 23.00\% & 19.25\% & 18.33\% & 18.62\% \\
$ES_3$ & 27.50\% & 26.50\% & 24.80\% & 27.38\%  & 28.00\% & 27.50\% & 30.00\% & 28.12\%\\
$ES_4$ & 26.00\% & 26.50\% & 30.50\% & 29.00\%  & 12.50\% & 21.50\% & 20.33\% & 20.52\%\\
$Cloud$& 3.50\% & 2.50\% & 3.00\% & 2.50\% & 8.50\% & 3.50\% & 1.67\% & 4.62\% \\
\hline
\end{tabular}%
\label{table:EUNumberOffloadingRate}
\end{table}

\begin{figure}
\centering
\subfigure[{WatDiv}]{%
        \centering
		\resizebox{0.45\columnwidth}{!}{
	\definecolor{lightBlue}{HTML}{42B4E7}
\definecolor{lightGrey}{HTML}{336699}
\definecolor{lightGreen}{HTML}{73B66B}
\definecolor{yellow}{HTML}{FFCB18}
\begin{tikzpicture}[]
    \begin{axis}[
        ybar,
        width=8cm,
        height=6cm,
        bar width=6pt,
        xlabel=Number Of Queries per End User,
        ylabel=Query Response Time (in ms),
        xtick=data,
        xticklabels={1,2,3,4},
        ymax=10000000,
        ymode=log,
        enlarge x limits=0.2,
        legend cell align=left,
        legend style={draw=none},
        legend pos=north west,
        legend columns=2 
    ]
    
    \addplot plot[ ybar, fill=lightBlue, draw=black, postaction={
        pattern=north east lines,
        pattern color=black
    }] coordinates {
        (5,     560313.9)
        (10,    2188321)
        (15,    11438976)
        (20,    43008456)
    };
    \addlegendentry{Cloud-Only};
    
    \addplot plot[ybar, fill=lightGrey, draw=black] coordinates {
        (5,     132180.3)
        (10,    2348216)
        (15,    2157195.5)
        (20,    2831496.75)
    };
    \addlegendentry{Random};
    
    \addplot plot[ybar, fill=lightGreen, draw=black, postaction={
        pattern=north west lines,
        pattern color=black 
    }] coordinates {
        (5,     19396.6)
        (10,    63575.2)
        (15,    122665.6)
        (20,    249328)
    };
    \addlegendentry{Edge-First};
    
    \addplot plot[ybar, fill=orange, draw=black] coordinates {
        (5,     19229.8)
        (10,      54461.6)
        (15,      111521.6)
        (20,      251447.5)
    };
    \addlegendentry{Greedy};
    
    \addplot plot[ybar, fill=yellow, draw=black] coordinates {
        (5,     17415.8)
        (10,    48699.8)
        (15,    104666)
        (20,    227437.2)
    };
    \addlegendentry{Ours};
 
    \end{axis}
\end{tikzpicture}
    }
       }
   \subfigure[{DBpedia}]{%
        \centering
		\resizebox{0.45\columnwidth}{!}{
	\definecolor{lightBlue}{HTML}{42B4E7}
\definecolor{lightGrey}{HTML}{336699}
\definecolor{lightGreen}{HTML}{73B66B}
\definecolor{yellow}{HTML}{FFCB18}
\begin{tikzpicture}[]
    \begin{axis}[
        ybar,
        width=8cm,
        height=6cm,
        bar width=6pt,
        xlabel=Number Of Queries per End User,
        ylabel=Query Response Time (in ms),
        xtick=data,
        xticklabels={1,2,3,4},
        ymax=10000000,
        ymode=log,
        enlarge x limits=0.2,
        legend cell align=left,
        legend style={draw=none},
        legend pos=north west,
        legend columns=2 
    ]
    
    \addplot plot[ ybar, fill=lightBlue, draw=black, postaction={
        pattern=north east lines,
        pattern color=black
    }] coordinates {
        (5,     1078871.7)
        (10,    69943366)
        (15,    69943366)
        (20,    69943366)
    };
    \addlegendentry{Cloud-Only};
    
    \addplot plot[ybar, fill=lightGrey, draw=black] coordinates {
        (5,     506738.3)
        (10,    8458211.5)
        (15,    1007069.8)
        (20,    9027337.7)
    };
    \addlegendentry{Random};
    
    \addplot plot[ ybar, fill=lightGreen, draw=black, postaction={
        pattern=north west lines,
        pattern color=black 
    }] coordinates {
        (5,     192703.3)
        (10,    527246)
        (15,    1033499.8)
        (20,    1768125.6)
    };
    \addlegendentry{Edge-First};
    
    \addplot plot[ybar, fill=orange, draw=black] coordinates {
        (5,     182857.8)
        (10,      519975.2)
        (15,      1007069.8)
        (20,      1725507.5)
    };
    \addlegendentry{Greedy};
    
    \addplot plot[ ybar, fill=yellow, draw=black] coordinates {
        (5,     174393.4)
        (10,    497529.6)
        (15,    967753.5)
        (20,    1669654.4)
    };
    \addlegendentry{Ours};
 
    \end{axis}
 
\end{tikzpicture}
    }
       }
{\caption{Varying Number of Queries per End User}}%
 \label{fig:ExpEndUserNumber}
\end{figure}
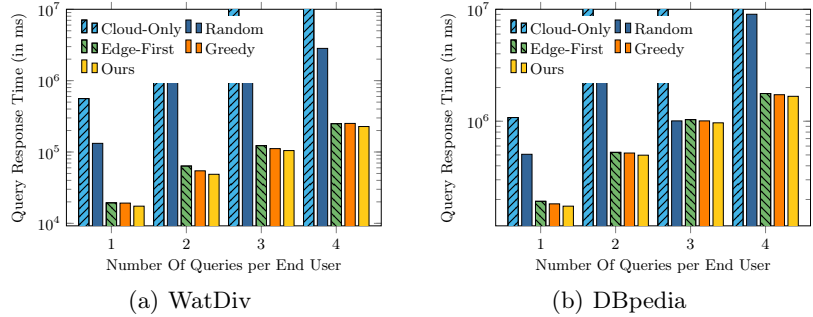

\nop{This experiment aims to explores the impact of varying numbers of end users while keeping other factors (e.g., edge server configuration and network bandwidth) constant. As the number of end users increases, the system needs to efficiently handle more queries to ensure user experience and system performance stability. We assume that the relationship between end users and queries remains linear, meaning that each additional end user corresponds to an additional query. 
Fig. \ref{fig:ExpEndUserNumber} and Table \ref{table:EndUserNumberOffloadingRate} show the query response time and assignment ratios under different numbers of end users. As the number of end users increases, the system needs to handle more query requests, which leads to a rise in response time. When the number of users is relatively small (e.g., 15), the edge servers can efficiently handle most queries, exhibiting a high assignment ratio to the edge servers. However, as the number of end users increases, the load on edge servers gradually increases, prompting the system to assign more queries to the cloud.

\begin{table}[]
\centering
\caption{{Assignment Ratio for Varying End Users}}
\begin{tabular}{rrrrrrrrr}
\hline
 & \multicolumn{4}{c}{WatDiv} & \multicolumn{4}{c}{DBpedia}\\
\cline{2-9}
& 15  & 20  & 25  & 30 & 15  & 20  & 25  & 30  \\
\hline
$ES_1$ & 16.00\% & 18.50\% & 18.00\% & 20.00\% & 30.00\% & 28.00\% & 23.20\% & 20.00\% \\
$ES_2$ & 29.33\% & 24.50\% & 24.80\% & 23.33\% & 18.00\% & 23.00\% & 19.60\% & 16.67\% \\
$ES_3$ & 26.00\% & 27.50\% & 26.00\% & 23.33\% & 27.33\% & 28.00\% & 26.00\% & 26.67\% \\
$ES_4$ & 28.00\% & 26.00\% & 25.20\% & 26.67\% & 18.00\% & 12.50\% & 16.80\% & 20.00\% \\
$Cloud$ &0.67\%  &  3.50\% &  6.00\% &  6.67\% &  6.67\% &  8.50\% & 14.40\% & 16.67\% \\
\hline
\end{tabular}%
\label{table:EndUserNumberOffloadingRate}
\end{table}

\begin{figure}
\centering
\subfigure[{WatDiv}]{%
        \centering
		\resizebox{0.45\columnwidth}{!}{
	\definecolor{lightBlue}{HTML}{42B4E7}
\definecolor{lightGrey}{HTML}{336699}
\definecolor{lightGreen}{HTML}{73B66B}
\definecolor{yellow}{HTML}{FFCB18}
\begin{tikzpicture}[font=\Large]
    \begin{axis}[
        ybar,
        bar width=6pt,
        xlabel=Number of End Users,
        ylabel=Query Response Time (in ms),
        xtick=data,
        xticklabels={15,20,25,30},
        ymax=10000000,
        ymode=log,
        enlarge x limits=0.15,
        legend cell align=left,
        legend style={draw=none},
        legend pos=north west,
        legend columns=2 
    ]
    
    \addplot plot[ ybar, fill=lightBlue, draw=black, postaction={
        pattern=north east lines,
        pattern color=black
    }] coordinates {
        (5,     363220.6)
        (10,    560313.9)
        (15,    1092300.7)
        (20,    1634647)
    };
    \addlegendentry{Cloud-Only};
    
    \addplot plot[ybar, fill=lightGrey, draw=black] coordinates {
        (5,     127852.7)
        (10,    132180.3)
        (15,    201920.4)
        (20,    282223.4)
    };
    \addlegendentry{Random};
    
    \addplot plot[ybar, fill=lightGreen, draw=black, postaction={
        pattern=north west lines,
        pattern color=black 
    }] coordinates {
        (5,     10971.7)
        (10,    19396.6)
        (15,    25505.8)
        (20,    38081.2)
    };
    \addlegendentry{Edge-First};
    
    \addplot plot[ybar, fill=orange, draw=black] coordinates {
        (5,     10612.5)
        (10,      19229.8)
        (15,      24481.2)
        (20,      35194.4)
    };
    \addlegendentry{Greedy};
    
    \addplot plot[ybar, fill=yellow, draw=black] coordinates {
        (5,     9431.8)
        (10,    17415.8)
        (15,    22941.3)
        (20,    28274.2)
    };
    \addlegendentry{Ours};
 
    \end{axis}
\end{tikzpicture}
    }
       \label{fig:ExpEndUserNumberWatDiv}%
       }
   \subfigure[{DBpedia}]{%
        \centering
		\resizebox{0.45\columnwidth}{!}{
	\definecolor{lightBlue}{HTML}{42B4E7}
\definecolor{lightGrey}{HTML}{336699}
\definecolor{lightGreen}{HTML}{73B66B}
\definecolor{yellow}{HTML}{FFCB18}
\begin{tikzpicture}[font=\Large]
    \begin{axis}[
        ybar,
        bar width=6pt,
        xlabel=Number Of End Users,
        ylabel=Query Response Time (in ms),
        xtick=data,
        xticklabels={15,20,25,30},
        ymax=10000000,
        ymode=log,
        enlarge x limits=0.15,
        legend cell align=left,
        legend style={draw=none},
        legend pos=north west,
        legend columns=2 
    ]
    
    \addplot plot[ ybar, fill=lightBlue, draw=black, postaction={
        pattern=north east lines,
        pattern color=black
    }] coordinates {
        (5,     992115.8)
        (10,    1078871.7)
        (15,    2551891)
        (20,    3455342)
    };
    \addlegendentry{Cloud-Only};
    
    \addplot plot[ybar, fill=lightGrey, draw=black] coordinates {
        (5,     374580.9)
        (10,    506738.3)
        (15,    962529.3)
        (20,    1113280.5)
    };
    \addlegendentry{Random};
    
    \addplot plot[ ybar, fill=lightGreen, draw=black, postaction={
        pattern=north west lines,
        pattern color=black 
    }] coordinates {
        (5,     124593.7)
        (10,    192703.3)
        (15,    288858.4)
        (20,    345918.6)
    };
    \addlegendentry{Edge-First};
    
    \addplot plot[ybar, fill=orange, draw=black] coordinates {
        (5,     125874.9)
        (10,      182857.8)
        (15,      270339.8)
        (20,      336957)
    };
    \addlegendentry{Greedy};
    
    \addplot plot[ ybar, fill=yellow, draw=black] coordinates {
        (5,     121860.8)
        (10,    174393.4)
        (15,    262882.8)
        (20,    321979)
    };
    \addlegendentry{Ours};
 
    \end{axis}
 
\end{tikzpicture}
    }
       \label{fig:ExpEndUserNumberDBpedia}%
       }
\caption{Varying Number of End Users}%
 \label{fig:ExpEndUserNumber}
\end{figure}
}

\subsubsection{Varying Selectivities of Queries}
This experiment aims to explore the impact of query selectivity on system response time and assignment ratio. Generally, the higher the query selectivity, the fewer results are returned. In the experiment, we analyze the impact by using queries with varying result sizes (in bytes).
Fig. \ref{fig:ExpQuerySelectivity} and Table \ref{table:ExpQuerySelectivity} show the changes in response time and assignment ratio under different selectivity conditions. The two datasets exhibit significant differences in result sizes, with Dbpedia showing more distinct selectivity levels. When the average result size is small (e.g., less than $10^{5}$), computation and transmission overhead is low, leading to shorter response times. In such cases, edge servers can efficiently process most queries, resulting in higher assignment ratios. As the average result size increases, the load on edge servers grows, leading to longer response times and lower assignment ratios. When the average result size is very large, although edge resources are limited, assigning queries to the cloud results in even higher transmission delays. Therefore, to optimize efficiency, more queries are assigned to edges. Our approach achieves up to a 37.41\% reduction in response time by effectively balancing computation and transmission costs.

\begin{table}[]
\centering
\caption{{Assignment Ratio for Varying Queries Selectivities}}
\begin{tabular}{rrrrrrrrr}
\hline
 & \multicolumn{4}{c}{WatDiv} & \multicolumn{4}{c}{DBpedia}\\
\cline{2-9}
& <$10^{5}$  & $10^{5}-10^{6}$  & $10^{6}-10^{7}$  & >$10^{7}$ & <$10^{5}$  & $10^{5}-10^{6}$  & $10^{6}-10^{8}$  & >$10^{8}$ \\
\hline
$ES_1$ & 15.50\% & 11.50\% & 29.50\% & 29.50\% & 27.00\% & 21.50\% & 28.50\% & 23.50\%\\
$ES_2$ & 26.00\% & 31.50\% & 24.00\% & 27.00\% & 28.00\% & 26.50\% & 24.00\% & 24.50\%\\
$ES_3$ & 26.50\% & 29.50\% & 16.00\% & 25.00\% & 23.00\% & 26.00\% & 28.50\% & 28.50\%\\
$ES_4$ & 21.00\% & 24.50\% & 26.00\% & 18.50\% & 22.00\% & 26.00\% & 16.00\% & 23.50\%\\
$Cloud$ &11.00\% & 3.00\%  &  4.50\% &  0.00\% &  0.00\% &  0.00\% &  3.00\% &  0.00\%\\
\hline
\end{tabular}%
\label{table:ExpQuerySelectivity}
\end{table}

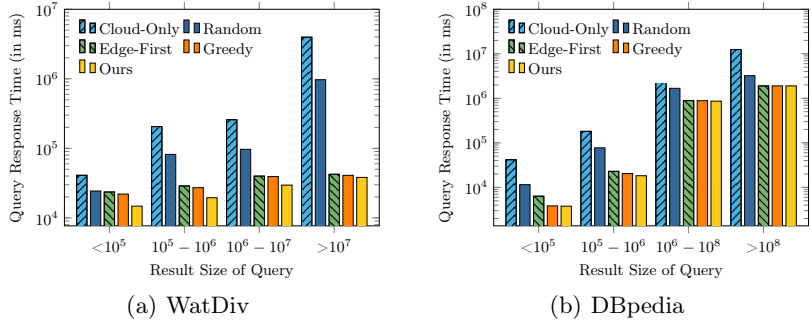
\begin{figure}
\centering
\subfigure[{WatDiv}]{%
        \centering 
		\resizebox{0.45\columnwidth}{!}{
	\definecolor{lightBlue}{HTML}{42B4E7}
\definecolor{lightGrey}{HTML}{336699}
\definecolor{lightGreen}{HTML}{73B66B}
\definecolor{yellow}{HTML}{FFCB18}
\begin{tikzpicture}[]
    \begin{axis}[
        ybar,
        width=8cm,
        height=6cm,
        bar width=6pt,
        xlabel=Result Size of Query,
        ylabel=Query Response Time (in ms),
        xtick=data,
        xticklabels={<$10^{5}$,$10^{5}-10^{6}$,$10^{6}-10^{7}$,>$10^{7}$},
        ymax=10000000,
        ymode=log,
        enlarge x limits=0.2,
        legend cell align=left,
        legend style={draw=none},
        legend pos=north west,
        legend columns=2 
    ]
    
    \addplot plot[ ybar, fill=lightBlue, draw=black, postaction={
        pattern=north east lines,
        pattern color=black
    }] coordinates {
        (5,     40979.8)
        (10,    205594)
        (15,    258078.7)
        (20,    3968884.5)
    };
    \addlegendentry{Cloud-Only};
    
    \addplot plot[ybar, fill=lightGrey, draw=black] coordinates {
        (5,     24366.2)
        (10,    81814.9)
        (15,    96881.1)
        (20,    971030.7)
    };
    \addlegendentry{Random};
    
    \addplot plot[ybar, fill=lightGreen, draw=black, postaction={
        pattern=north west lines,
        pattern color=black 
    }] coordinates {
        (5,     23574.6)
        (10,    28799)
        (15,    39883)
        (20,    42368.1)
    };
    \addlegendentry{Edge-First};
    
    \addplot plot[ybar, fill=orange, draw=black] coordinates {
        (5,     22020)
        (10,      27212)
        (15,      39284)
        (20,      40944.1)
    };
    \addlegendentry{Greedy};
    
    \addplot plot[ybar, fill=yellow, draw=black] coordinates {
        (5,     14754.5)
        (10,    19512)
        (15,    29613)
        (20,    38217.4)
    };
    \addlegendentry{Ours};
 
    \end{axis}
\end{tikzpicture}
    }
       \label{fig:ExpQuerySelectivityWatDiv}%
       }
   \subfigure[{DBpedia}]{%
        \centering 
		\resizebox{0.45\columnwidth}{!}{
	\definecolor{lightBlue}{HTML}{42B4E7}
\definecolor{lightGrey}{HTML}{336699}
\definecolor{lightGreen}{HTML}{73B66B}
\definecolor{yellow}{HTML}{FFCB18}
\begin{tikzpicture}[]
    \begin{axis}[
        ybar,
        width=8cm,
        height=6cm,
        bar width=6pt,
        xlabel=Result Size of Query,
        ylabel=Query Response Time (in ms),
        xtick=data,
        xticklabels={<$10^{5}$,$10^{5}-10^{6}$,$10^{6}-10^{8}$,>$10^{8}$},
        ymax=100000000,
        ymode=log,
        enlarge x limits=0.2,
        legend cell align=left,
        legend style={draw=none},
        legend pos=north west,
        legend columns=2 
    ]
    
    \addplot plot[ ybar, fill=lightBlue, draw=black, postaction={
        pattern=north east lines,
        pattern color=black
    }] coordinates {
        (5,     41780.8)
        (10,    181227.7)
        (15,    5182022)
        (20,    12464541)
    };
    \addlegendentry{Cloud-Only};
    
    \addplot plot[ybar, fill=lightGrey, draw=black] coordinates {
        (5,     11499.4)
        (10,    77373.8)
        (15,    1678310.6)
        (20,    3227753.6)
    };
    \addlegendentry{Random};
    
    \addplot plot[ ybar, fill=lightGreen, draw=black, postaction={
        pattern=north west lines,
        pattern color=black 
    }] coordinates {
        (5,     6339)
        (10,    22905)
        (15,    891067.3)
        (20,    1901883.6)
    };
    \addlegendentry{Edge-First};
    
    \addplot plot[ybar, fill=orange, draw=black] coordinates {
        (5,     3838.5)
        (10,    20437.9)
        (15,    885366.5)
        (20,    1901883.6)
    };
    \addlegendentry{Greedy};
    
    \addplot plot[ ybar, fill=yellow, draw=black] coordinates {
        (5,     3788)
        (10,    18274.7)
        (15,    865727.6)
        (20,    1901883.6)
    };
    \addlegendentry{Ours};
 
    \end{axis}
\end{tikzpicture}
    }
       \label{fig:ExpQuerySelectivityDBpedia}%
       }
\caption{Varying Selectivities of Queries}%
 \label{fig:ExpQuerySelectivity}
\end{figure}

{\subsection{Overhead Analysis}}

{\subsubsection{Scheduling overhead analysis.}
Fig. 14 shows the scheduling time and its proportion under different datasets and system scales.
The results show that the scheduling time increases as expected with system scale.
On computation-intensive DBpedia the proportion of scheduling time remains consistently low, while queries in the WatDiv are relatively lightweight. Thus, the complexity of the algorithm becomes more evident under large-scale settings.
Nevertheless, even in this case, the scheduling overhead does not become the dominant factor in system performance.}

\begin{figure}
\centering
\subfigure[{WatDiv}]{%
        \centering
		\resizebox{0.45\columnwidth}{!}{


\definecolor{lightBlue}{HTML}{42B4E7}
\definecolor{lightGrey}{HTML}{336699}
\definecolor{lightGreen}{HTML}{73B66B}
\definecolor{yellow}{HTML}{FFCB18}
\begin{tikzpicture}
\begin{axis}[
    ybar,
    width=8cm,
    height=6cm,
    bar width=8pt,
    xtick=data,
    xlabel={($|\mathcal{K}|, |\mathcal{N}|$)},
    xticklabel style={rotate=30,anchor=east,yshift=-2mm},
    xticklabels={(4,20), (8,40), (16,80), (32,160), (64,320), (128,640)},
    ylabel={Time (in ms)},
    ymode=log,
    ymin=1000, ymax=500000,
    legend style={at={(0.05,0.95)}, anchor=north west, draw=none},
    axis y line*=left, 
]

\addplot[ybar, fill=yellow] coordinates {
    (1, 1286) (2, 3597) (3, 8834) (4, 23906) (5, 87416) (6, 324259)
};
\addlegendentry{Scheduling Time}
\end{axis}

\begin{axis}[
    width=8cm,
    height=6cm,
    axis y line*=right,
    axis x line=none,
    ylabel={Ratio (\%)},
    ymin=0, 
    ymax=100,
    legend style={at={(0.05,0.85)}, anchor=north west, draw=none},
]

\addplot[lightGrey, thick, mark=triangle*, mark size=3pt] coordinates {
    (1, 7.38) (2, 9.00) (3, 10.63) (4, 16.60) (5, 25.68) (6, 33.68)
};
\addlegendentry{Scheduling Overhead Ratio}
\end{axis}
\end{tikzpicture}
    }
       }
   \subfigure[{DBpedia}]{%
        \centering
		\resizebox{0.45\columnwidth}{!}{
	\definecolor{lightBlue}{HTML}{42B4E7}
\definecolor{lightGrey}{HTML}{336699}
\definecolor{lightGreen}{HTML}{73B66B}
\definecolor{yellow}{HTML}{FFCB18}
\begin{tikzpicture}
\begin{axis}[
    ybar,
    width=8cm,
    height=6cm,
    bar width=8pt,
    xtick=data,
    xlabel={($|\mathcal{K}|, |\mathcal{N}|$)},
    xticklabel style={rotate=30,anchor=east,yshift=-2mm},
    xticklabels={(4,20), (8,40), (16,80), (32,160), (64,320), (128,640)},
    ylabel={Time (in ms)},
    ymode=log,
    ymin=1000, ymax=500000,
    legend style={at={(0.05,0.95)}, anchor=north west, draw=none},
    axis y line*=left, 
]

\addplot[fill=yellow] coordinates {
    (1, 2121) (2, 4369) (3, 14352) (4, 30546) (5, 99437) (6, 375029)
};
\addlegendentry{Scheduling Time}

\end{axis}

\begin{axis}[
    width=8cm,
    height=6cm,
    axis y line*=right,
    axis x line=none,
    ylabel={Ratio (\%)},
    ymin=0, 
    ymax=20,
    legend style={at={(0.05,0.85)}, anchor=north west, draw=none},
]

\addplot[lightGrey, thick, mark=triangle*, mark size=3pt] coordinates {
    (1, 1.12) (2, 0.76) (3, 0.58) (4, 0.66) (5, 1.32) (6, 1.99)
};
\addlegendentry{Scheduling Overhead Ratio}

\end{axis}
\end{tikzpicture}
    }
       }
{\caption{Scheduling Overhead Ratio}}%
 \label{fig:ExpScheduling}
\end{figure}
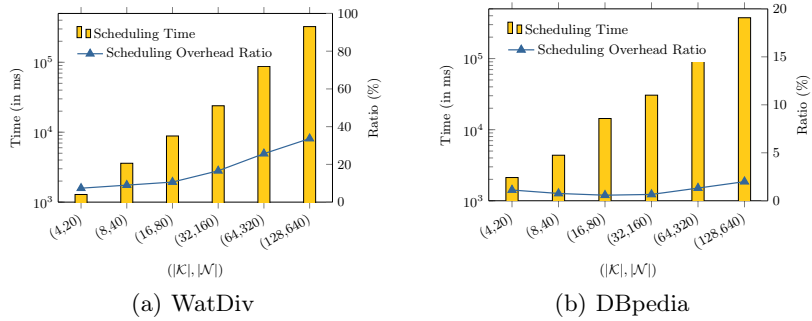

{\subsubsection{Construction overhead.}
We report, in Table 11, the construction time for WatDiv and DBpedia, across varying numbers of edge servers ($|\mathcal{K}|$) and end users ($|\mathcal{N}|$). 
The construction time scales near-linearly with the number of servers, demonstrating the scalability of our approach for large-scale deployments. Since this process is performed offline, it does not incur any latency for online query processing.
}

{
\begin{table}[]
\centering
{\caption{{Pattern-induced subgraph construction overhead (sec)}}}
\begin{tabular}{c@{\hspace{0.5em}}c@{\hspace{0.5em}}c@{\hspace{0.5em}}c@{\hspace{0.5em}}c@{\hspace{0.5em}}c@{\hspace{0.5em}}c}
\hline
($|\mathcal{K}|, |\mathcal{N}|$) & (4, 20) & (8, 40) & (16, 80) & (32, 160) & (64, 320) & (128, 640)\\
\hline
WatDiv & 11.34 & 19.36 & 35.42 & 66.89 & 122.97 & 246.22\\
DBpedia & 58.84 & 78.45 & 104.62 & 174.74 & 323.17 & 591.17\\
\hline                  
\end{tabular}%
\label{table:construction_time}
\end{table}}

\vspace{-0.2in}
\section{Related Work}\label{sec:relatedwork}

\noindent\textbf{Database Systems in Edge Computing.} Several database systems have been tailored for the edge computing environment \cite{VLDB2020:SuccinctEdge, EDBT2021:SuccinctEdge, SIGMOD18:DPaxos, SoCC19:WedgeDB, HotEdge20:DataEdge} to cope with limited resources and latency sensitivity. SuccinctEdge \cite{VLDB2020:SuccinctEdge, EDBT2021:SuccinctEdge} extends LiteMat \cite{ESWC19:LiteMat} with a compact, self-indexed in-memory layout for storing RDF graphs on edge servers.
In terms of data consistency,
DPaxos \cite{SIGMOD18:DPaxos} introduces a Paxos-based consensus protocol for managing partitioned data across datacenters and edge nodes.
WedgeDB \cite{SoCC19:WedgeDB} enables distributed and secure transaction processing in edge environments.
Trivedi et al. \cite{HotEdge20:DataEdge} further analyze data sharing requirements in dege applications and propose Griffin, an edge data sharing service.

\noindent\textbf{Task Offloading in Edge Computing.} The task offloading allocation problem in edge computing has attracted significant attention \cite{TVT19:JointTaskOffloading, TMC19:EnergyEfficientOffloading,TMC21:MultiTaskLearning,TPDS22:RevenueDriven,TSC23:UnknownSystemSide,TMC23:MultiObjOffloading,TMC24:GameTheoryApproach}. 
In mobile edge computing systems that only edge-side computation, Tran et al. \cite{TVT19:JointTaskOffloading} and Yang et al. \cite{TMC21:MultiTaskLearning} formulated the problem as a mixed-integer nonlinear programming, while Ma et al. \cite{TPDS22:RevenueDriven} applied convex optimization and heuristic algorithm. 
For edge-cloud system, Guo et al. \cite{TMC19:EnergyEfficientOffloading} proposed a mixed-integer programming problem and distributed solution algorithm, and Chen et al. \cite{TMC24:GameTheoryApproach} introduced a game-theoretic decentralized algorithm.
For other system scenarios,
Dai et al. \cite{TSC23:UnknownSystemSide} employed the multi-armed bandit theory \cite{10.1023/A:1013689704352} to develop an online learning-based dependent task offloading algorithm, and Xiao et al. \cite{TMC23:MultiObjOffloading} optimized D2D-assisted offloading via enhanced binary particle swarm algorithm.


{\section{Conclusions and Future Work}}\label{sec:Conclusions} 
In this paper, we study the assignment of SPARQL query processing to edge servers to improve system efficiency.
We introduce pattern-induced subgraphs to guide data placement on edge servers to quickly identify servers capable of handling specific queries.
For network scheduling, we model the joint query offloading decision and resource allocation problem as an MINLP problem and propose an assignment algorithm.
Experimental results show that our method achieves higher computational efficiency than baseline methods. 
Future work will extend this framework to support complex queries by combining multiple patterns, drawing on techniques from view materialization and query rewriting \cite{DBLP:journals/pacmmod/CaoGZZ25,DBLP:journals/pacmmod/Pang0YY24,DBLP:conf/sigmod/TroullinouKLM21}. We aim to explore inter-edge communication to further enhance flexibility and load balancing among edge servers.

\section{AI-Generated Content Acknowledgment} 
We have not employed any artificial intelligence (AI) techniques to produce any content in this paper.

\bibliographystyle{splncs04}
\bibliography{sample-base}

@String{Computing = "Computing" }

@String{Computer = "{IEEE} Computer" }

@String{Academic = "Academic Press" }

@String{Springer = "Springer-Verlag" }

@article{TMC19:EnergyEfficientOffloading,
  author       = {Songtao Guo and
                  Jiadi Liu and
                  Yuanyuan Yang and
                  Bin Xiao and
                  Zhetao Li},
  title        = {Energy-Efficient Dynamic Computation Offloading and Cooperative Task
                  Scheduling in Mobile Cloud Computing},
  journal      = {IEEE Trans. Mob. Comput.},
  volume       = {18},
  number       = {2},
  pages        = {319--333},
  year         = {2019}
}

@article{TMC21:MultiTaskLearning,
  author       = {Bo Yang and
                  Xuelin Cao and
                  Joshua Bassey and
                  Xiangfang Li and
                  Lijun Qian},
  title        = {Computation Offloading in Multi-Access Edge Computing: {A} Multi-Task
                  Learning Approach},
  journal      = {IEEE Trans. Mob. Comput.},
  volume       = {20},
  number       = {9},
  pages        = {2745--2762},
  year         = {2021}
}

@inproceedings{EDBT2021:SuccinctEdge,
  author       = {Weiqin Xu and
                  Olivier Cur{\'{e}} and
                  Philippe Calvez},
  title        = {{Knowledge Graph Management on the Edge}},
  booktitle    = {{EDBT}},
  pages        = {229--240},
  publisher    = {OpenProceedings.org},
  year         = {2021}
}

@article{VLDB2020:SuccinctEdge,
  author       = {Weiqin Xu and
                  Olivier Cur{\'{e}} and
                  Philippe Calvez},
  title        = {{SuccinctEdge: A Succinct {RDF} Store for Edge Computing}},
  journal      = {Proc. {VLDB} Endow.},
  volume       = {13},
  number       = {12},
  pages        = {2857--2860},
  year         = {2020}
}

@inproceedings{DBLP:conf/sigmod/Hartig13,
  author       = {Olaf Hartig},
  title        = {{SQUIN:} a traversal based query execution system for the web of linked
                  data},
  booktitle    = {Proceedings of the {ACM} {SIGMOD} International Conference on Management
                  of Data, {SIGMOD} 2013, New York, NY, USA, June 22-27, 2013},
  pages        = {1081--1084},
  publisher    = {{ACM}},
  year         = {2013}
}

@article{DBLP:journals/ws/VerborghSHHVMHC16,
  author       = {Ruben Verborgh and others},
  title        = {Triple Pattern Fragments: {A} low-cost knowledge graph interface for
                  the Web},
  journal      = {J. Web Semant.},
  volume       = {37-38},
  pages        = {184--206},
  year         = {2016},
}

@inproceedings{DBLP:conf/semweb/HartigLP17,
  author       = {Olaf Hartig and
                  Ian Letter and
                  Jorge P{\'{e}}rez},
  title        = {A Formal Framework for Comparing Linked Data Fragments},
  booktitle    = {The Semantic Web - {ISWC} 2017 - 16th International Semantic Web Conference},
  series       = {Lecture Notes in Computer Science},
  volume       = {10587},
  pages        = {364--382},
  publisher    = {Springer},
  year         = {2017},
}

@ARTICLE{8318578,
  author={Liu, Liqing and Chang, Zheng and Guo, Xijuan},
  journal={IEEE Internet of Things Journal}, 
  title={Socially Aware Dynamic Computation Offloading Scheme for Fog Computing System With Energy Harvesting Devices}, 
  year={2018},
  volume={5},
  number={3},
  pages={1869-1879}
}

@inproceedings{DBLP:conf/icc/FengZDCY18,
  author       = {Jie Feng and
                  Liqiang Zhao and
                  Jianbo Du and
                  Xiaoli Chu and
                  F. Richard Yu},
  title        = {{Energy-Efficient Resource Allocation in Fog Computing Supported IoT with Min-Max Fairness Guarantees}},
  booktitle    = {{ICC}},
  pages        = {1--6},
  publisher    = {{IEEE}},
  year         = {2018}
}

@book{book:MobileBroadband,
title = {4G LTE/LTE-Advanced for Mobile Broadband},
author = {Erik Dahlman and Stefan Parkvall and Johan Sköld},
publisher = {Academic Press},
address = {Oxford},
pages = {iv},
year = {2011},
isbn = {978-0-12-385489-6}
}

@article{Dataset:DBpedia,
  author       = {Jens Lehmann and othersr},
  title        = {{DBpedia - A large-scale, multilingual knowledge base extracted from
                  Wikipedia}},
  journal      = {Semantic Web},
  volume       = {6},
  number       = {2},
  pages        = {167--195},
  year         = {2015}
}

@inproceedings{DBLP:conf/semweb/AlucHOD14,
  author       = {G{\"{u}}nes Alu{\c{c}} and
                  Olaf Hartig and
                  M. Tamer {\"{O}}zsu and
                  Khuzaima Daudjee},
  title        = {Diversified Stress Testing of {RDF} Data Management Systems},
  booktitle    = {The Semantic Web - {ISWC} 2014 - 13th International Semantic Web Conference},
  series       = {Lecture Notes in Computer Science},
  volume       = {8796},
  pages        = {197--212},
  year         = {2014},
}

@ARTICLE {TKDE2021:FMQO,
author = {P. Peng and Q. Ge and L. Zou and M. Ozsu and Z. Xu and D. Zhao},
journal = {IEEE Trans. Knowl. Data Eng.},
title = {Optimizing Multi-Query Evaluation in Federated RDF Systems},
year = {2021},
volume = {33},
number = {04},
issn = {1558-2191},
pages = {1692-1707},
publisher = {IEEE Computer Society},
address = {Los Alamitos, CA, USA},
month = {apr}
}

@InProceedings{DASFAA2018:FMQO,
author="Peng, Peng
and Zou, Lei
and {\"O}zsu, M. Tamer
and Zhao, Dongyan",
title="Multi-query Optimization in Federated RDF Systems",
booktitle="Database Systems for Advanced Applications",
year="2018",
publisher="Springer International Publishing",
address="Cham",
pages="745--765",
isbn="978-3-319-91452-7"
}

@inproceedings{10.1145/1367497.1367578,
author = {Stocker, Markus and Seaborne, Andy and Bernstein, Abraham and Kiefer, Christoph and Reynolds, Dave},
title = {SPARQL Basic Graph Pattern Optimization Using Selectivity Estimation},
year = {2008},
isbn = {9781605580852},
publisher = {Association for Computing Machinery},
booktitle = {WWW},
pages = {595–604},
numpages = {10}
}

@inproceedings{10.1145/1559845.1559911,
author = {Neumann, Thomas and Weikum, Gerhard},
title = {Scalable Join Processing on Very Large RDF Graphs},
year = {2009},
isbn = {9781605585512},
publisher = {Association for Computing Machinery},
address = {New York, NY, USA},
booktitle = {SIGMOD},
pages = {627–640},
numpages = {14}
}

@inproceedings{Dataset:QALD,
  author       = {Aleksandr Perevalov and
                  Dennis Diefenbach and
                  Ricardo Usbeck and
                  Andreas Both},
  title        = {QALD-9-plus: {A} Multilingual Dataset for Question Answering over
                  DBpedia and Wikidata Translated by Native Speakers},
  booktitle    = {{ICSC}},
  pages        = {229--234},
  publisher    = {{IEEE}},
  year         = {2022}
}

@article{DBLP:journals/actanum/BelottiKLLLM13,
  author       = {Pietro Belotti and
                  Christian Kirches and
                  Sven Leyffer and
                  Jeff T. Linderoth and
                  James R. Luedtke and
                  Ashutosh Mahajan},
  title        = {Mixed-integer nonlinear optimization},
  journal      = {Acta Numer.},
  volume       = {22},
  pages        = {1--131},
  year         = {2013}
}

@inproceedings{Tammer1987TheAO,
  title={The application of parametric optimization and imbedding to the foundation and realization of a generalized primal decomposition approach},
  author={Klaus Tammer},
  year={1987}
}

@article{TVT19:JointTaskOffloading,
  author       = {Tuyen X. Tran and
                  Dario Pompili},
  title        = {Joint Task Offloading and Resource Allocation for Multi-Server Mobile-Edge
                  Computing Networks},
  journal      = {IEEE Transactions on Vehicular Technology},
  volume       = {68},
  number       = {1},
  pages        = {856--868},
  year         = {2019}
}

@ARTICLE{9861697,
  author={Xiao, Zhu and Shu, Jinmei and Jiang, Hongbo and Lui, John C. S. and Min, Geyong and Liu, Jiangchuan and Dustdar, Schahram},
  journal={IEEE Trans. Mob. Comput.}, 
  title={Multi-Objective Parallel Task Offloading and Content Caching in D2D-aided MEC Networks}, 
  year={2022},
  volume={},
  number={},
  pages={1-16}
}

@article{VLDBJ2014:gStore,
author = {Zou, Lei and \"{O}zsu, M. Tamer and Chen, Lei and Shen, Xuchuan and Huang, Ruizhe and Zhao, Dongyan},
title = {{GStore: A Graph-Based SPARQL Query Engine}},
year = {2014},
publisher = {Springer-Verlag},
address = {Berlin, Heidelberg},
volume = {23},
number = {4},
issn = {1066-8888},
journal = {The VLDB Journal},
pages = {565–590},
numpages = {26}
}

@article{VLDB2011:GStore,
author = {Zou, Lei and Mo, Jinghui and Chen, Lei and \"{O}zsu, M. Tamer and Zhao, Dongyan},
title = {{GStore: Answering SPARQL Queries via Subgraph Matching}},
year = {2011},
publisher = {VLDB Endowment},
volume = {4},
number = {8},
issn = {2150-8097},
journal = {Proc. VLDB Endow.},
pages = {482–493},
numpages = {12}
}

@ARTICLE {TMC23:MultiObjOffloading,
author = {Zhu Xiao  and others},
journal = {IEEE Trans. Mob. Comput.},
title = {Multi-Objective Parallel Task Offloading and Content Caching in D2D-Aided MEC Networks},
year = {2023},
volume = {22},
number = {11},
issn = {1558-0660},
pages = {6599-6615},
publisher = {IEEE Computer Society},
address = {Los Alamitos, CA, USA},
}

@ARTICLE {TSC23:UnknownSystemSide,
author = {Xingxia Dai and  
Zhu Xiao and  
Hongbo Jiang and  
Ming Lei and 
Geyong Min and  
Jiangchuan Liu and  
Schahram Dustdar},
journal = {IEEE Transactions on Services Computing},
title = {Offloading Dependent Tasks in Edge Computing With Unknown System-Side Information},
year = {2023},
volume = {16},
number = {06},
issn = {1939-1374},
pages = {4345-4359},
publisher = {IEEE Computer Society},
address = {Los Alamitos, CA, USA}
}

@ARTICLE {TMC24:GameTheoryApproach,
author = {Ying Chen and  
Jie Zhao and
Yuan Wu and
Jiwei Huang and
Xuemin Shen},
journal = {IEEE Trans. Mob. Comput.},
title = {QoE-Aware Decentralized Task Offloading and Resource Allocation for End-Edge-Cloud Systems: A Game-Theoretical Approach},
year = {2024},
volume = {23},
number = {01},
issn = {1558-0660},
pages = {769-784},
publisher = {IEEE Computer Society},
address = {Los Alamitos, CA, USA},
}

@article{TPDS22:RevenueDriven,
author = {Ma, Zhi and others},
title = {Towards Revenue-Driven Multi-User Online Task Offloading in Edge Computing},
year = {2022},
issue_date = {May 2022},
publisher = {IEEE Press},
volume = {33},
number = {5},
issn = {1045-9219},
journal = {IEEE Transactions on Parallel and Distributed Systems},
pages = {1185–1198},
numpages = {14}
}

@inproceedings{ESWC19:LiteMat,
author = {Cur\'{e}, Olivier and Xu, Weiqin and Naacke, Hubert and Calvez, Philippe},
title = {LiteMat, an Encoding Scheme with RDFS++ and Multiple Inheritance Support},
year = {2019},
publisher = {Springer-Verlag},
address = {Berlin, Heidelberg},
booktitle = {ESWC},
pages = {269–284},
numpages = {16}
}

@inproceedings{SoCC19:WedgeDB,
  author       = {Abhishek A. Singh and
                  Faisal Nawab},
  title        = {WedgeDB: Transaction Processing for Edge Databases},
  booktitle    = {SoCC},
  pages        = {482},
  publisher    = {{ACM}},
  year         = {2019}
}

@inproceedings{SIGMOD18:DPaxos,
author = {Nawab, Faisal and Agrawal, Divyakant and El Abbadi, Amr},
title = {DPaxos: Managing Data Closer to Users for Low-Latency and Mobile Applications},
year = {2018},
publisher = {Association for Computing Machinery},
booktitle = {SIGMOD},
pages = {1221–1236},
numpages = {16}
}

@inproceedings{ISWC2018:Neptune,
  author       = {Bradley R. Bebee and others},
  title        = {Amazon Neptune: Graph Data Management in the Cloud},
  booktitle    = {ISWC},
  volume       = {2180},
  publisher    = {CEUR-WS.org},
  year         = {2018},
}

@inproceedings{HotEdge20:DataEdge,
  author       = {Animesh Trivedi and
                  Lin Wang and
                  Henri E. Bal and
                  Alexandru Iosup},
  title        = {Sharing and Caring of Data at the Edge},
  booktitle    = {HotEdge},
  publisher    = {{USENIX} Association},
  year         = {2020}
}

@article{10.1023/A:1013689704352,
author = {Auer, Peter and Cesa-Bianchi, Nicol\`{o} and Fischer, Paul},
title = {Finite-time Analysis of the Multiarmed Bandit Problem},
year = {2002},
issue_date = {May-June 2002},
publisher = {Kluwer Academic Publishers},
address = {USA},
volume = {47},
number = {2–3},
issn = {0885-6125},
journal = {Mach. Learn.},
month = may,
pages = {235–256},
numpages = {22}
}

@article{DBLP:journals/access/ZhaoTQN17,
  author       = {Pengtao Zhao and
                  Hui Tian and
                  Cheng Qin and
                  Gaofeng Nie},
  title        = {Energy-Saving Offloading by Jointly Allocating Radio and Computational
                  Resources for Mobile Edge Computing},
  journal      = {{IEEE} Access},
  volume       = {5},
  pages        = {11255--11268},
  year         = {2017}
}

@inproceedings{SIGMOD04:gIndex,
author = {Yan, Xifeng and Yu, Philip S. and Han, Jiawei},
title = {Graph indexing: a frequent structure-based approach},
year = {2004},
isbn = {1581138598},
publisher = {Association for Computing Machinery},
address = {New York, NY, USA},
booktitle = {SIGMOD},
pages = {335–346},
numpages = {12}
}

@article{TON2019:OnDisc,
  author       = {Zhenhua Han and
                  Haisheng Tan and
                  Xiang{-}Yang Li and
                  Shaofeng H.{-}C. Jiang and
                  Yupeng Li and
                  Francis C. M. Lau},
  title        = {OnDisc: Online Latency-Sensitive Job Dispatching and Scheduling in
                  Heterogeneous Edge-Clouds},
  journal      = {{IEEE/ACM} Trans. Netw.},
  volume       = {27},
  number       = {6},
  pages        = {2472--2485},
  year         = {2019},
}

@article{TMC2024:QoEAware,
  author       = {Ying Chen and
                  Jie Zhao and
                  Yuan Wu and
                  Jiwei Huang and
                  Xuemin Shen},
  title        = {QoE-Aware Decentralized Task Offloading and Resource Allocation for
                  End-Edge-Cloud Systems: {A} Game-Theoretical Approach},
  journal      = {{IEEE} Trans. Mob. Comput.},
  volume       = {23},
  number       = {1},
  pages        = {769--784},
  year         = {2024},
}

@article{TMC2024:GameTheoretic,
  author       = {Liantao Wu and
                  Peng Sun and
                  Zhibo Wang and
                  Yanjun Li and
                  Yang Yang},
  title        = {Computation Offloading in Multi-Cell Networks With Collaborative Edge-Cloud
                  Computing: {A} Game Theoretic Approach},
  journal      = {{IEEE} Trans. Mob. Comput.},
  volume       = {23},
  number       = {3},
  pages        = {2093--2106},
  year         = {2024},
}

@inproceedings{10.1145/2588555.2594535,
author = {Papailiou, Nikolaos and others},
title = {H2RDF+: an efficient data management system for big RDF graphs},
year = {2014},
address = {New York, NY, USA},
booktitle = {Proceedings of the 2014 ACM SIGMOD International Conference on Management of Data},
pages = {909–912},
numpages = {4},
}

@article{DBLP:journals/fcsc/Ozsu16,
  author       = {M. Tamer {\"{O}}zsu},
  title        = {A survey of {RDF} data management systems},
  journal      = {Frontiers Comput. Sci.},
  volume       = {10},
  number       = {3},
  pages        = {418--432},
  year         = {2016},
}

@article{DBLP:journals/tkde/HusainMMKT11,
  author       = {Mohammad Farhan Husain and
                  James P. McGlothlin and
                  Mohammad M. Masud and
                  Latifur R. Khan and
                  Bhavani M. Thuraisingham},
  title        = {Heuristics-Based Query Processing for Large {RDF} Graphs Using Cloud
                  Computing},
  journal      = {{IEEE} Trans. Knowl. Data Eng.},
  volume       = {23},
  number       = {9},
  pages        = {1312--1327},
  year         = {2011},
}

@article{DBLP:journals/vldb/KaoudiM15,
  author       = {Zoi Kaoudi and
                  Ioana Manolescu},
  title        = {{RDF} in the clouds: a survey},
  journal      = {{VLDB} J.},
  volume       = {24},
  number       = {1},
  pages        = {67--91},
  year         = {2015},
}

@article{DBLP:journals/vldb/AliSYHN22,
  author       = {Waqas Ali and
                  Muhammad Saleem and
                  Bin Yao and
                  Aidan Hogan and
                  Axel{-}Cyrille Ngonga Ngomo},
  title        = {A survey of {RDF} stores {\&} {SPARQL} engines for querying knowledge
                  graphs},
  journal      = {{VLDB} J.},
  volume       = {31},
  number       = {3},
  pages        = {1--26},
  year         = {2022},
}

@article{10.14778/3151106.3151109,
author = {Abdelaziz, Ibrahim and Harbi, Razen and Khayyat, Zuhair and Kalnis, Panos},
title = {A survey and experimental comparison of distributed SPARQL engines for very large RDF data},
year = {2017},
issue_date = {September 2017},
publisher = {VLDB Endowment},
volume = {10},
number = {13},
}

@ARTICLE{8454442,
  author={Ning, Zhaolong and Dong, Peiran and Kong, Xiangjie and Xia, Feng},
  journal={IEEE Internet of Things Journal}, 
  title={A Cooperative Partial Computation Offloading Scheme for Mobile Edge Computing Enabled Internet of Things}, 
  year={2019},
  volume={6},
  number={3},
  pages={4804-4814}
}

@ARTICLE{9340353,
  author={Liu, Lina and Sun, Bo and Wu, Yuan and Tsang, Danny H. K.},
  journal={IEEE Internet of Things Journal}, 
  title={Latency Optimization for Computation Offloading With Hybrid NOMA–OMA Transmission}, 
  year={2021},
  volume={8},
  number={8},
  pages={6677-6691}
}

@ARTICLE{8314696,
  author={Chen, Min and Hao, Yixue},
  journal={IEEE Journal on Selected Areas in Communications}, 
  title={Task Offloading for Mobile Edge Computing in Software Defined Ultra-Dense Network}, 
  year={2018},
  volume={36},
  number={3},
  pages={587-597}
}

@article{DiscreteMathematics:HomomorphismsProperties,
title = {Homomorphisms and colourings of oriented graphs: An updated survey},
journal = {Discrete Mathematics},
volume = {339},
number = {7},
pages = {1993-2005},
year = {2016},
note = {7th Cracow Conference on Graph Theory, Rytro 2014},
issn = {0012-365X},
author= {{\'{E}}ric Sopena},
}

@article{DBLP:journals/pacmmod/CaoGZZ25,
  author       = {Lisheng Cao and
                  Xiangyang Gou and
                  Lei Zou and
                  Wenjie Zhang},
  title        = {{MAVIS:} Materialized View for Subgraph Matching},
  journal      = {Proc. {ACM} Manag. Data},
  volume       = {3},
  number       = {6},
  pages        = {1--26},
  year         = {2025}
}

@article{DBLP:journals/pacmmod/Pang0YY24,
  author       = {Yue Pang and
                  Lei Zou and
                  Jeffrey Xu Yu and
                  Linglin Yang},
  title        = {Materialized View Selection {\&} View-Based Query Planning for
                  Regular Path Queries},
  journal      = {Proc. {ACM} Manag. Data},
  volume       = {2},
  number       = {3},
  pages        = {152},
  year         = {2024}
}

@inproceedings{DBLP:conf/sigmod/TroullinouKLM21,
  author       = {Georgia Troullinou and
                  Haridimos Kondylakis and
                  Matteo Lissandrini and
                  Davide Mottin},
  title        = {{SOFOS:} Demonstrating the Challenges of Materialized View Selection
                  on Knowledge Graphs},
  booktitle    = {{SIGMOD} '21: International Conference on Management of Data, Virtual
                  Event, China, June 20-25, 2021},
  pages        = {2789--2793},
  publisher    = {{ACM}},
  year         = {2021}
}

@ARTICLE{10133894,
  author={Gong, Taiyuan and Zhu, Li and Yu, F. Richard and Tang, Tao},
  journal={IEEE Transactions on Intelligent Transportation Systems}, 
  title={Edge Intelligence in Intelligent Transportation Systems: A Survey}, 
  year={2023},
  volume={24},
  number={9},
  pages={8919-8944}
}

@INPROCEEDINGS{5767845,
  author={Xiao, Xiaokui and Yao, Bin and Li, Feifei},
  booktitle={2011 IEEE 27th International Conference on Data Engineering}, 
  title={Optimal location queries in road network databases}, 
  year={2011},
  volume={},
  number={},
  pages={804-815}
}

@article{qiu2020edge,
  title={Edge computing in industrial internet of things: Architecture, advances and challenges},
  author={Qiu, Tie and Chi, Jiancheng and Zhou, Xiaobo and Ning, Zhaolong and Atiquzzaman, Mohammed and Wu, Dapeng Oliver},
  journal={IEEE communications surveys \& tutorials},
  volume={22},
  number={4},
  pages={2462--2488},
  year={2020}
}

@article{LEPHUOC201625,
title = {The Graph of Things: A step towards the Live Knowledge Graph of connected things},
journal = {Journal of Web Semantics},
volume = {37-38},
pages = {25-35},
year = {2016},
issn = {1570-8268},
author = {Danh Le-Phuoc and Hoan {Nguyen Mau Quoc} and Hung {Ngo Quoc} and Tuan {Tran Nhat} and Manfred Hauswirth}
}

@ARTICLE{8289317,
  author={Premsankar, Gopika and Di Francesco, Mario and Taleb, Tarik},
  journal={IEEE Internet of Things Journal}, 
  title={Edge Computing for the Internet of Things: A Case Study}, 
  year={2018},
  volume={5},
  number={2},
  pages={1275-1284},
}

@article{belotti2009branching,
  title={Branching and bounds tighteningtechniques for non-convex MINLP},
  author={Belotti, Pietro and Lee, Jon and Liberti, Leo and Margot, Fran{\c{c}}ois and W{\"a}chter, Andreas},
  journal={Optimization Methods \& Software},
  volume={24},
  number={4-5},
  pages={597--634},
  year={2009},
  publisher={Taylor \& Francis}
}


\end{document}